 \documentclass[UKenglish,cleveref, autoref, thm-restate]{lipics-v2021}
\usepackage{cite}
\usepackage{amsmath,amssymb,amsfonts}

\usepackage{algorithm}
\usepackage[noend]{algpseudocode}
\usepackage{comment}
\usepackage{graphicx}
\usepackage{textcomp}
\usepackage{xcolor}
\usepackage{url}

\usepackage{array}
\usepackage{ragged2e}
\usepackage{longtable}
\usepackage{booktabs}
\usepackage{tcolorbox}
\usepackage{amsthm}
\usepackage{hyperref}

\usepackage{tikz}
\usetikzlibrary{calc,cd,hobby}

\newcommand{\ol}[1]{\overline{#1}}

\DeclareMathOperator{\rank}{rank}

\DeclareMathOperator{\height}{height}
\DeclareMathOperator{\subtree}{subtree}

\DeclareMathOperator{\LCA}{LCA}

\DeclareMathOperator{\reg}{reg}
\DeclareMathOperator{\Reg}{Reg}

\DeclareMathOperator{\Ham}{Ham}
\DeclareMathOperator{\Rem}{Rem}
\DeclareMathOperator{\Exc}{Exc}
\DeclareMathOperator{\Par}{Par}
\DeclareMathOperator{\Max}{Max}
\DeclareMathOperator{\Block}{Block}
\DeclareMathOperator{\HBlock}{HBlock}
\DeclareMathOperator{\nMulti}{nMulti}
\DeclareMathOperator{\HMulti}{HMulti}
\DeclareMathOperator{\nSRuns}{nSRuns}
\DeclareMathOperator{\vSRuns}{vSRuns}
\DeclareMathOperator{\vRuns}{vRuns}
\DeclareMathOperator{\HSRuns}{HSRuns}
\DeclareMathOperator{\DS}{DS}
\DeclareMathOperator{\HDS}{HDS}
\DeclareMathOperator{\Inv}{Inv}
\DeclareMathOperator{\nRuns}{nRuns}
\DeclareMathOperator{\HRuns}{HRuns}
\DeclareMathOperator{\hist}{hist}
\DeclareMathOperator{\loc}{loc}

\DeclareMathOperator{\Hist}{Hist}
\DeclareMathOperator{\Loc}{Loc}
\DeclareMathOperator{\nSUS}{nSUS}
\DeclareMathOperator{\vSUS}{vSUS}
\DeclareMathOperator{\vLRM}{vLRM}
\DeclareMathOperator{\HSUS}{HSUS}
\DeclareMathOperator{\nSMS}{nSMS}
\DeclareMathOperator{\vSMS}{vSMS}
\DeclareMathOperator{\HSMS}{HSMS}
\DeclareMathOperator{\osc}{osc}

\DeclareMathOperator{\Syn}{Syn}
\DeclareMathOperator{\age}{age}
\DeclareMathOperator{\unimst}{unimst}

  \tikzset{notestyleraw/.append style={rectangle}}
 \author{Ivor van der Hoog}{IT University of Copenhagen, Denmark}{ivva@itu.dk}{
 https://orcid.org/0009-0006-2624-0231}{}

 \author{Eva Rotenberg}{IT University of Copenhagen, Denmark}{erot@itu.dk}{
 https://orcid.org/0000-0001-5853-7909}{}

 \author{Daniel Rutschmann}{Institute of Science and Technology Austria (ISTA), Austria}{rutschmann@proton.me}{ https://orcid.org/0009-0005-6838-2628}{}

\def\BibTeX{{\rm B\kern-.05em{\sc i\kern-.025em b}\kern-.08em
    T\kern-.1667em\lower.7ex\hbox{E}\kern-.125emX}}
\begin{document}

\title{Adaptive Sorting, also in the BST Model
}

\titlerunning{Adaptive sorting, also in the BST model}

\authorrunning{Ivor van der Hoog, Eva Rotenberg, and Daniel Rutschmann}

\Copyright{Ivor van der Hoog, Eva Rotenberg, and Daniel Rutschmann}

 \funding{{\it Ivor van der Hoog}, {\it Eva Rotenberg}, and {\it Daniel Rutschmann} are grateful to the VILLUM Foundation for supporting this work via Eva Rotenberg's  grant (VIL37507) ``Efficient Recomputations for Changeful Problems''. This work was supported by the European Union's Horizon 2020 research and innovation programme under the Marie Sk\l{}odowska-Curie grant agreement No 899987.}

\hideLIPIcs

\nolinenumbers


\setcounter{page}{0}


\maketitle

\begin{abstract}
Sorting is one of the most basic algorithmic problems.
In adaptive sorting, algorithms perform better on inputs that are easier to sort.
Formally, one defines a measure of sortedness that captures features of the input, typically an entropy measure, and tries to design algorithms that adapt to this measure.
Over the years, many such measures have been studied, raising two natural questions:
How do these measures relate to each other, and are there algorithms that are adaptive to each measure?

We answer these questions by completing a hierarchy of relations between previously studied entropy measures.
Our added relations show that this hierarchy has a source, $\Reg$, which is a classical measure from the literature.
We then \textbf{\emph{discover}} and \emph{prove} that an existing data structure already yields a sorting algorithm adaptive to $\Reg$.
Consequently, this algorithm already achieves the adaptive bounds of all other measures in the hierarchy.

We then make a brief mathematical detour.
McIlroy previously noted that $\Reg$ has a mathematically inconvenient definition.
We argue for studying a slightly weaker, but mathematically much more elegant measure: $\reg$.
We show that, if the input contains no duplicates, $O(n + \log \reg) = O(n + \log \Reg)$. This makes $\reg$ elegant, easier to analyse, and almost as strong.

Finally, we turn to algorithm design.
We introduce \textsc{SproutSort}, a $\reg$-adaptive algorithm in the offline binary search tree model (the \emph{BST Model}).
In this model, a single binary search tree stores the distinct input keys, which must be accessed in input order.
The entire access sequence is known in advance, but each access starts at the root and incurs unit cost for pointer moves and rotations.
\textsc{SproutSort} achieves an access cost of $O(n+\log\reg)$, making it the most adaptive algorithm in the BST Model to date.
Turning our attention back to adaptive sorting in the comparison model, we also present the adaptive sorting algorithm \textsc{SkeletonSort}.
This is a simple randomised pointer-machine algorithm that sorts distinct keys in expected $O(n+\log\reg)$ time.
Although it does not work if the input has duplicates, \textsc{SkeletonSort} offers a considerably simpler implementation of an algorithm this high in the adaptive hierarchy.
\end{abstract}

\newpage

\section{Introduction}

This paper studies comparison-based adaptive sorting of $n$ elements and the related offline binary search tree (BST) access problem.
Sorting $n$ distinct elements requires $\Omega(n\log n)$ comparisons in the worst case.
However, this lower bound does not take into account the fact that some inputs are easier to solve than others.
If the input is already sorted, one can verify this in linear time and, intuitively, an \emph{almost sorted input} can also be sorted faster.  
A large body of work studies \emph{adaptive sorting algorithms} that exploit input features 
that reflect notions of almost-sortedness, expressed by entropy functions. 
Examples include inputs consisting of few sorted contiguous pieces~\cite{Knuth1998,Mannila1985}, few contiguous blocks~\cite{CarlssonLevcopoulosPetersson1990}, or few (possibly interleaved) increasing or monotone subsequences~\cite{LevcopoulosPetersson1994}. 
Whenever a measure of entropy is introduced, a natural question is how it relates to other measures: is one strictly stronger than the other, in the sense that matching one in running time, implies matching the other? 
In this paper, we extend the hierarchy of entropy measures by establishing the implications marked in orange in Figure~\ref{fig:hierarchy}.
As a consequence, the hierarchy has a source $\Reg$.

For each entropy measure in the hierarchy, one can ask whether there exists algorithms that are adaptive to that measure. 
We \textbf{\emph{discover}} that the unified bound dictionary of B\u{a}doiu, Cole, Demaine, and Iacono~\cite{BadoiuColeDemaineIacono2007} already yields a $\Reg$-adaptive algorithm.
Thus, \cite{BadoiuColeDemaineIacono2007} yields the most adaptive sorting algorithms to date. 
We continue by modifying $\Reg$ to define a simpler measure $\reg$ with much nicer mathematical properties.
We show that $\reg$ dominates almost all prior measures. 
Moreover, if the input contains no duplicates then we prove that any $\reg$-adaptive algorithm is also $\Reg$-adaptive.
We give a $\reg$-adaptive offline BST access algorithm. Finally, we present a simple $\reg$-adaptive classical algorithm.

\subparagraph{Terminology.}
To explain our results, 
consider the following abstractions of adaptive sorting bounds:
let $\mathcal I_n$ denote the set of all size-$n$ inputs.
A \emph{notion of presortedness} assigns to every input $X \in \mathcal I_n$ some target cost $f(X)$.
A sorting algorithm is \emph{comparison-adaptive} with respect to this notion if it sorts every input $X$ in $O(f(X))$ comparisons, and \emph{time-adaptive} if its total running time is $O(f(X))$. 
For offline BST algorithms, we consider the BST model where one can navigate and rotate nodes of a binary search tree, and we consider the total access cost of accessing all elements of $X$ in-order.
The target $f$ is usually defined in two steps: first a number $k(X)$ is defined in a sophisticated way,
and then $f(X)$ is a function of $n$ and $k(X)$, e.g. $n (1+\log k(X))$ or $n (1+k(X))$.
A concrete example is the number of runs~\cite{Knuth1998}: By splitting $X$ between every pair of adjacent items that form a decreasing sequence of length 2, we partition $X$ into maximal contiguous non-decreasing sequences $X_1, X_2, \dots, 
X_{k=\nRuns}$.
Then, \cite{Knuth1998} puts $f(X) = n (1+ \log (\nRuns))$ and presents a (comparison) adaptive algorithm.

\subparagraph{Ranks and Duality.}
 For a list $X$ of elements and $x \in X$ we define the rank of $x$ denoted by $\operatorname{rank}_{X}(x)$ as its index in the sorted order, \emph{after} removing duplicates.
Consider for the sake of analysis, simultaneously replacing each element $X[i]$ by $\operatorname{rank}_{X}(X[i])$.
Then any comparison-based algorithm makes the exact same decisions for the replaced input. Thus, we may assume that $X$ is a list of elements in $[n]$.
If $X$ contains no duplicates then $X$ is a permutation. If $X$ is a permutation, it has a unique inverse $X^{-1}$ which we call its dual, and which leads to  
dual notions of presortedness.
For example, dual runs: $\nRuns^*(X) := \nRuns(X^{-1})$, giving rise to $f(X) = n (1+\log (\nRuns^*(X)))$.
In most cases, these dual definitions can be (manually) extended to inputs $X$ with duplicates.
E.g., A subsequence $X_i \subseteq X$ is \emph{value-contiguous} if for all $x, y, z \in X$ with $x \le y \le z$ and $x, z \in X_i$, we also have $y \in X_i$.
We can partition $X = X_1 \sqcup \dots \sqcup X_k$ into maximal non-decreasing value-contiguous sequences $X_i$
and put $\nRuns^*(X) := k$; this agrees with the previous definition when $X$ contains no duplicates.
We \emph{occasionally} define duality for lists with duplicates, by considering the position of the $i$'th element in the sorted order after a stable sort, denoted $\ol{X}[i]$.
Note that $\ol{X}$ is a permutation.

\subparagraph{Entropy functions.} The Shannon entropy function allows us to define strong notions of presortedness from a given partition of $X$.
We illustrate three examples:
For runs, consider the partition $X = X_1 \sqcup \dots \sqcup X_k$ into maximal contiguous non-decreasing sequences,
put $\HRuns := \sum_{i=1}^{k} \frac{|X_i|}{n} \log_2 \frac{n}{|X_i|}$ and let $f_{\HRuns}(X) = n (1+\HRuns(X))$.
For dual runs, consider the partition $X = X_1 \sqcup \dots \sqcup X_k$ into maximal \emph{value-contiguous} non-decreasing sequences,
put $\HRuns^* := \sum_{i=1}^{k} \frac{|X_i|}{n} \log_2 \frac{n}{|X_i|}$ and let $f_{\HRuns^*}(X) = n (1+\HRuns^*(X))$.
McIlroy~\cite{McIlroy1993} shows that merge sort with exponential search is both $\HRuns$-adaptive and $\HRuns^*$-adaptive.
For multi-set sorting, assume $X \subseteq [n]^n$ and denote for $i \in [n]$ let $m_i$ be the number of occurrences of $i$.
Put $\HMulti := \sum_{i=1}^{n} \frac{m_i}{n} \log \frac{n}{m_i}$ and $f_{\HMulti}(X) := n(1+\HMulti(X))$.
\cite{MunroSpira1976} shows that merge sort, with a minor adaption, is $\HMulti$-adaptive.

{
\subparagraph{A hierarchy.}
Measures of presortedness form a hierarchy.
We draw an arrow from a measure $M$ to a measure $N$, and write $M \preceq N$, if $f_M(X) \in O(f_N(X))$ for all inputs $X$.
In other words, any $M$-adaptive algorithm is automatically also $N$-adaptive.
For example, all occurrences of the same value in $X$ form a value-contiguous subsequence of $X$, and each dual run is a union of such sequences.
Consequently, $\HRuns^*(X)\leq \HMulti(X)$, so an algorithm being $\HRuns^*$-adaptive implies adaptivity for multi-set sorting.
Figure~\ref{fig:hierarchy} illustrates the hierarchy of the measures considered here. 
Its source is $\Reg$.
Here, we explain the definition of $\Reg$ which is complicated, and considers for each element $i$ its rank in $X[1..i]$.
\[
\Reg(X):=\prod_{i=2}^n\min_{j<i}
\left\{(i-j) \cdot \max\left\{1,
\left|\operatorname{rank}_{X[1..i]}(X[i])
-\operatorname{rank}_{X[1..i]}(X[j])\right|\right\}\right\}.
\]
The measure $\Reg$ is derived from insertion sort: process $X$ from left to right and insert $X[i]$ into $X[1..(i-1)]$ using a good \emph{finger} $X[j]$.
A good finger is close to $X[i]$ both in terms of index and rank.
Crucially, and inconveniently (as acknowledged by McIlroy \cite{McIlroy1993}), these ranks are taken only among the elements seen so far.
For example, for $X=(1,5,4,3,2)$ consider $i = 2, j = 1$: when $X[2]=5$ is inserted, only $X[1] = 1$ and $X[2] = 5$ have been seen, so the rank difference is $1$ and the cost is $(2-1)\lvert2-1\rvert=1$, even though $|X[2] - X[1]| = 4$.
Thus, $\Reg$ considers this insertion to be cheap as $2,3,4$ have not yet been inserted.
Unfortunately, this makes $\Reg$ not invariant under reversing $X$.
Our alternative, weaker, but more natural definition $\reg$ simply considers for $X[i]$ and $X[j]$ their ranks in the full input $X$; the target cost $n+\log\reg(X)$ is invariant up to constant factors under reversing the sequence or its value order and, for permutations, under inversion: $n+\log\reg(X)\in\Theta(n+\log\reg(X^{-1}))$.
\[
\reg(X):=\prod_{i=2}^n\min_{j<i}
\left\{(i-j) \cdot \max\left\{1,
\left|\rank_X(X[i])
-\rank_X (X[j])\right|\right\}\right\}.
\]
}

\subparagraph{Problem statements.}
Our input is a list $X$ of $n$ elements from a totally ordered universe, which may contain duplicates.
The first is comparison-based sorting: compute a non-decreasing ordering of $X$ using only comparisons to determine the relative order of elements.
The second is offline BST access: given a BST containing exactly the distinct values of $X$, access $X[1],X[2],\dots,X[n]$ in input order (see Section~\ref{sub:bst_def} for the formal definition of the BST model). The goal is to obtain adaptive bounds for both problems.

\subparagraph{Contributions.}
Our first contribution, is that we establish the hierarchy in Figure~\ref{fig:hierarchy}, adding several new or previously unpublished relations, establishing that $\Reg$ is its source. Moffat and Petersson~\cite{MoffatPetersson1991} achieved an algorithm that used $\Reg$-adaptive comparisons, but whose running time has an additive $O(n\log\log n)$ time overhead. We \textbf{\emph{discover}} that insertion sort using the unified bound dictionary of B\u{a}doiu, Cole, Demaine, and Iacono~\cite{BadoiuColeDemaineIacono2007} takes $O(n+\log\Reg(X))$ time, eliminating this overhead.
We then make a case for studying the mathematically more convenient measure $\reg$:

 This measure is near-identical to $\Reg$, except that it considers rank distances in the \emph{full input.} 
We show the interesting mathematical property that $n+\log\reg(X)$ is, up to constant factors, the weight of a minimum spanning tree with edge weights $\log(2+|i-j|+|\rank_X(X[i])-\rank_X(X[j])|)$. This characterisation explains the symmetric properties of $\reg$. Moreover, $\reg$ dominates every measure in Figure~\ref{fig:hierarchy} except the measures that consider prefix ranks $\Reg$, $\Hist$, and $\Loc$.
Finally, we show that if $X$ contains no duplicates then $O(n + \log \reg) = O(n + \log \Reg)$, which illustrates the strength of $\reg$ (see \Cref{theo:Reg_reg}).

Our first new algorithm, \textsc{SproutSort}, addresses offline BST access.
Given the rank of each element and its multiplicity, it implements a BST merge sort of cost $O(n+\log\reg(X))$. Through the framework of Blelloch and Dobson~\cite{BlellochDobson2023} and our extension to repeated keys, it yields a deterministic offline BST access algorithm with the same cost. Existing unified bound dictionaries and heaps lie outside the BST model, while known BST algorithms incur an additive $O(n\log\log n)$ overhead~\cite{DerryberrySleator2009,DemaineEtAl2013}. The $\Reg$-adaptive algorithms can first compute the ranks and multiplicities~\cite{BadoiuColeDemaineIacono2007}, and so our access algorithm runs in $O(n+\log\reg(X))$ time.

Finally, we present \textsc{SkeletonSort}, a simple randomised pointer-machine sorting algorithm. It sorts distinct comparable keys in expected $O(n+\log\reg(X))$ time without prior knowledge of their ranks.
Since $O(n + \log \reg(X)) = O(n + \log \Reg(X))$, this algorithm is as good as 
~\cite{BadoiuColeDemaineIacono2007} when the input has no duplicates.

\begin{figure}[H]
    \centering
\includegraphics[]{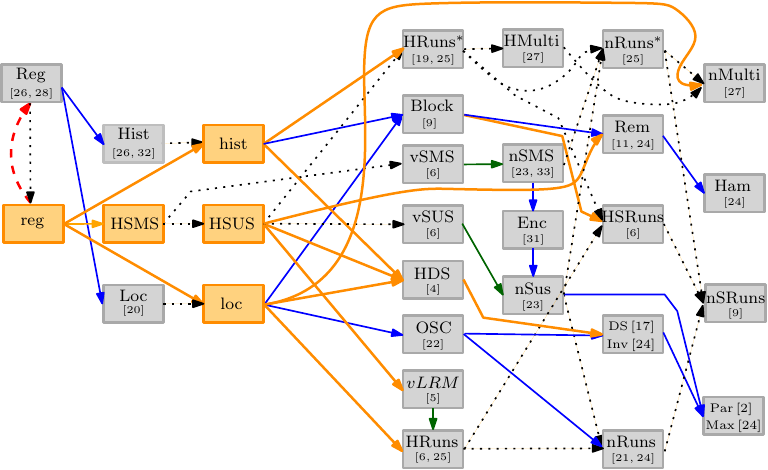}
\caption{Orange\&red are established here, red only holds when the input has no duplicates. Blue is from~\cite{PeterssonMoffat1995}, green from~\cite{BarbayNavarro2013}, and black is straightforward. For entropy measure definitions, see Table~\ref{tab:presortedness}.
We move the base of the blue arrows~\cite{PeterssonMoffat1995} from their original Hist and Loc to hist and loc, although these were not defined at the time of \cite{PeterssonMoffat1995}, since it follows from the essence of their proof.} \label{fig:hierarchy}
\end{figure}

\newpage

\section{Discovery} \label{sect:discovery}

In this section, we discover that $\Reg$ corresponds to the unified bound for dictionaries~\cite{BadoiuColeDemaineIacono2007} and establish the hierarchy in Figure~\ref{fig:short_slim}.
In Appendix~\ref{sect:arrows} we extend the hierarchy so that $\Reg$ dominates all other previous measures in Figure~\ref{fig:hierarchy}. 
Thus, we \emph{discover} that the dictionary from~\cite{BadoiuColeDemaineIacono2007} yields an insertion-sort algorithm that dominates all previously studied adaptive sorting algorithms from~\cite{MunroSpira1976,Mannila1985,Skiena1988,KatajainenLevcopoulosPetersson1989,CarlssonLevcopoulosPetersson1990,MoffatPetersson1991,McIlroy1993,LevcopoulosPetersson1993,LevcopoulosPetersson1994,BarbayFischer2012,BarbayNavarro2013,BlellochDobson2023}.
We next revisit the observation of McIlroy~\cite{McIlroy1993} that $\Reg$ is mathematically inelegant.
We make the case for a mathematically more convenient and elegant (but slightly weaker) measure $\reg$, and show that this measure also suffices to dominate all measures in the hierarchy except for three, all equally mathematically inconvenient, measures listed in Section~\ref{sect:reg_relations}.
In Appendix~\ref{sect:BST_loc_hist}, we further explore the definitions of presortedness measures and explain how the classical $\Hist$ and $\Loc$ from Figure~\ref{fig:hierarchy} correspond to the working set and dynamic finger properties for binary search trees, respectively.
For completeness, Appendix~\ref{app:synergystic} discusses the measure by Barbay, Ochoa and Satti~\cite{BarbayOchoaSatti2017}.

\subparagraph{Discovering a $\Reg$-adaptive algorithm.}
Badoiu, Cole, Demaine and Iacono~\cite[Theorem 3]{BadoiuColeDemaineIacono2007} design a \emph{unified bound} dictionary.
We show that sorting with their dictionary is $\Reg$-adaptive.

\begin{definition}
A dictionary maintains a set $Y$ of distinct elements with
multiplicities. Inserting $x$ adds it to $Y$ if absent, and
otherwise increases its multiplicity. The dictionary can
report $Y$ in sorted order in $O(|Y|)$ time.
Insertion sort inserts all elements of $X$, and then reports $Y$
in sorted order, repeating each element according to its
multiplicity.
Let $\age(y)$ count the distinct elements accessed since the last access to $y$. The first insertion takes $O(1)$ time. Thereafter, the
\emph{unified bound} demands that inserting $x$ takes amortized
$
O\left(\min_{y\in Y}
\left\{
\log(2+|\rank_{Y\cup\{x\}}(x)-\rank_{Y\cup\{x\}}(y)|)
+\log(2+\age(y))
\right\}\right)
$
time.
\end{definition}

\begin{theorem}
Insertion sort with a unified bound dictionary sorts a sequence $X$ of $n$ elements in $O(n+\log\Reg(X))$ time and comparisons.
\end{theorem}

\begin{proof}
We insert $X$ into a unified bound
dictionary. Processing
$X[i]$ consists of inserting it if absent, or searching for
it and increasing its multiplicity otherwise.

For every $j<i$, the last access to $X[j]$ occurred at
or after step $j$, so $\age(X[j])\leq i-j$.
Since $\rank$ ignores duplicates, and $Y\cup\{X[i]\}$
and $X[1..i]$ contain the same values, we have
$\rank_{Y\cup\{X[i]\}}(\cdot)=\rank_{X[1..i]}(\cdot)$.
Thus, the unified bound gives a total running time of
\begin{equation*}
\begin{split}
O(n+\sum_{i=2}^n\min_{j<i}\left\{
\log(2+|\rank_{X[1..i]}(X[i])-\rank_{X[1..i]}(X[j])|)
+\log(2+i-j)\right\})
 = O(n+\log\Reg(X)).
\end{split}
\end{equation*}
Finally, we output the distinct elements in sorted order,
repeating each according to its multiplicity, in $O(n)$ time.
The number of comparisons is bounded by the running time.
\end{proof}

\begin{figure}[H]
    \centering
    \includegraphics[width = 0.92\linewidth]{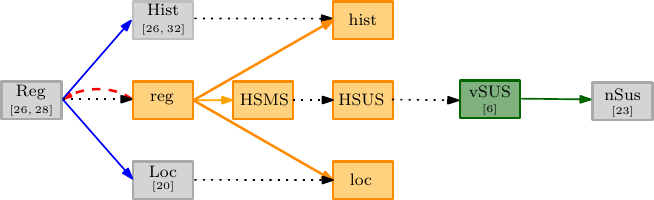}
    \caption{The relations established in this section. 
    Orange and red are new, green is~\cite{BarbayNavarro2013}.}
    \label{fig:short_slim}
\end{figure}

\subsection{The Case for reg} \label{sect:reg_mst}

The measure $\Reg$ has a mathematically inconvenient property:
reversing $X$ may significantly change its value.
For example, let $n=k^2$ with $k\geq2$, put
$S=(1,k+1,\ldots,(k-1)k+1)$, and let
$X=(S,\ldots,S,1,2,\ldots,n)$, with $k$ copies of $S$.
Thus, $X$ has length $2n$.
In $X$, consecutive items have prefix-rank distance at most $1$,
except at the start of each subsequent copy of $S$ and at the
start of the final increasing run, where a previous occurrence
of the same value is at index distance $k$.
Thus, $\Reg(X)\leq k^k=2^{o(n)}$.

However, reversing $X$ places $(n,n-1,\ldots,1)$ first,
followed by $k$ reversed copies of $S$.
Consider the $n-k$ accesses after the first reversed copy of $S$.
Every finger in the initial descending run has index distance
at least $k$.
Among the reversed copies of $S$, different values have
prefix-rank distance at least $k$, since all values have already
appeared, while equal values have index distance at least $k$.
Thus, each of these $n-k$ accesses contributes a factor at
least $k$ to $\Reg(\mathrm{Reverse}(X))$.
Conversely, choosing the preceding item as a finger bounds
each factor by $n$, and the initial descending run contributes
only factors of $1$.
Hence
$\Reg(\mathrm{Reverse}(X))=k^{\Theta(n)}=2^{\Theta(n\log n)}$.

Consequently, there are two different, incomparable variants
of $\Reg$ which we find not particularly elegant.
For this reason, we wish to make the case for the simpler measure
\[
\reg(X):=\prod_{i=2}^n\min_{j<i}
\left\{(i-j)\cdot\max\left\{1,
|X[i]-X[j]|\right\}\right\}.
\]
In fact, McIlroy~\cite{McIlroy1993} already notes that $\Reg$ is
needlessly complicated compared to $\reg$
(although he incorrectly claims that
$\log\reg \leq \log\Reg + O(n)$, see above counterexample).
Although $\reg$ is weaker than $\Reg$, 
it appears to capture the essence of both measures.
We establish another property that we believe illustrates the
mathematical elegance of $\reg$:
the bound $n+\log\reg$ is equivalent to the weight of a
minimum spanning tree in a natural graph.
This characterisation makes $O(n+\log\reg)$ invariant up to constant
factors under reversing $X$, reversing its values,
or replacing $X$ by $X^{-1}$. We show:

\begin{theorem}\label{theo:reg_mst}
Given $X$ of length $n\geq2$, consider the complete graph $G(X)$ on $[n]$, with
edge weights $w(\{i,j\})=\log(2+|i-j|+|X[i]-X[j]|)$.
Let $\unimst(X)$ be the weight of a minimum spanning tree of $G(X)$.
Then $n+\log\reg(X)\in\Theta(\unimst(X))$.
\end{theorem}

\begin{definition}
A \emph{finger function} is a map
$f:[n]\setminus\{1\}\to[n]$ with $f(i)<i$.
Recall that $\reg(X)$ minimises
$\prod_{i=2}^n (i-f(i))\max\{1,|X[i]-X[f(i)]|\}$
over all finger functions $f$.
\end{definition}

\begin{lemma}\label{lemm:forest_star}
Every forest with $m$ edges contains a spanning subgraph $H$
with at least $m/2$ edges, such that every connected component
of $H$ is a star.
\end{lemma}
\begin{proof}
Root each tree arbitrarily and partition its edges according
to whether the child endpoint has even or odd depth.
In either part, every vertex is adjacent either only to its
parent or only to its children.
Thus, every connected component is a star. Since the two parts partition the $m$ edges,
one contains at least $m/2$ edges.
\end{proof}

\begin{proof}[Proof of \cref{theo:reg_mst}]
Every finger function $f$ defines a spanning tree, since
$f(i)<i$ for all $i>1$.
Moreover, $w(\{i,f(i)\})$ is at most
$2+\log(i-f(i))+\log\max\{1,|X[i]-X[f(i)]|\}$.
Taking the minimum over all finger functions gives
$\unimst(X)\leq 2n+\log\reg(X)$.

Conversely, let $T$ be a minimum spanning tree.
We construct a finger function $f$ whose edges have total
weight $O(w(T)+n)$.

Let $T'$ be a copy of $T$.
Initially, every vertex forms its own $0$-cluster.
In rounds $r=1,2,\ldots$, contract all $(r-1)$-clusters
and consider the forest $F$ of edges whose current weight
lies in $[2r-2,2r)$.
By \cref{lemm:forest_star}, choose a star forest
$H\subseteq F$ containing at least half its edges.
Merge the clusters of each star into an $r$-cluster,
and increase the weight of every edge in $F-H$ by $2$.
Clusters not involved in a merge remain unchanged.

We first bound the diameter of the clusters.
Put $d(i,j)=2^{w(i,j)}$.
By the triangle inequality,
$d(i,j)\leq d(i,k)+d(k,j)$.
Let $D_r$ be the maximum $d$-diameter of an $r$-cluster.
Since every star has diameter realized by at most two edges, each
of weight less than $2r$, we have
$ D_0=2$ and $D_r\leq 3D_{r-1}+2\cdot 2^{2r}$. 
Induction gives $D_r\leq 2^{2r+3}$.
Thus, any two vertices in the same $r$-cluster
have weight at most $2r+3$.

The \emph{leader} of a cluster is its smallest-index vertex.
Whenever an $r$-cluster merges $k$ smaller clusters with leaders
$i_1<\cdots<i_k$, we set $f(i_j)=i_1$ for every $j>1$.
This is valid since $i_1<i_j$.
Moreover, each vertex receives a finger exactly when it ceases
to be a leader. Since eventually all vertices belong to one
cluster, whose leader is $1$, every vertex except $1$
receives exactly one finger.
These $k-1$ finger edges have weight at most $2r+3$,
while the corresponding star contains $k-1$ edges
of weight at least $2r-2$.
Summing over all merges gives $
w(f)\leq w(T')+5(n-1)$.

Finally, every round uses at least as many edges as it
increases in weight.
Since each edge is used in exactly one merge,
the total weight increase is at most $2(n-1)$.
Hence $w(T')\leq w(T)+2(n-1)$, and therefore
$w(f)\in O(w(T)+n)$.
Since $\log\reg(X)\leq 2w(f)$ and $w(T)\geq n-1$,
we conclude that $n+\log\reg(X)\in O(\unimst(X))$.
\end{proof}

\subsection{Clean definitions for existing measures: loc, hist, HSMS, HSUS}
In addition to $\Reg$, Moffat and Petersson~\cite{MoffatPetersson1991}
define $\Hist$ and $\Loc$ (see Table~\ref{tab:presortedness}),
which lie high in the hierarchy of Figure~\ref{fig:hierarchy}.
Both measures have been studied through adaptive sorting algorithms~\cite{MoffatPetersson1991,BlellochDobson2023}.
Like $\Reg$, they use prefix ranks and thus share its mathematical inconveniences.
McIlroy~\cite{McIlroy1993} briefly mentioned more elegant alternatives, analogous to $\reg$:
\begin{align*}
\loc(X) = \prod_{i=2}^{n} \max\{1, |X[i] - X[i-1]|\} \quad\text{and}\quad
\hist(X) = \prod_{i=2}^{n} |\ol{X}^{-1}[i] - \ol{X}^{-1}[i-1]|.
\end{align*}

We can define mathematically convenient counterparts for more existing measures:
Levcopoulos and Petersson~\cite{LevcopoulosPetersson1994}
study partitions $X=X_1\sqcup\dots\sqcup X_k$ into non-decreasing subsequences.
As the $X_i$ need not be contiguous or value-contiguous,
this generalises both runs and dual runs.
They define $\nSUS$ as the minimum number $k$ of parts,
and give a sorting algorithm with adaptive bound
$f_{\nSUS}(X)=n(1+\log\nSUS(X))$.
Barbay and Navarro~\cite{BarbayNavarro2013} strengthen this bound
by considering the entropy of the part sizes rather than their number.
They compute a partition with $\nSUS$ parts and define
$H(\vSUS)$ as the entropy of its part sizes.
Since any partition into $k$ parts has entropy at most $\log k$,
this strengthens the $\nSUS$ bound.
However, as they observe, their partition need not minimize entropy.
We therefore define
\[
\HSUS(X)=
\min_{X=X_1\sqcup\dots\sqcup X_k}
\sum_{i=1}^k \frac{|X_i|}{n}\log\frac{n}{|X_i|}
\]
where the minimum is over all partitions into non-decreasing
subsequences.
Clearly, $\HSUS\preceq H(\vSUS)$, giving a stronger measure.
We further generalize these partitions by allowing each $X_i$
to be either non-decreasing or non-increasing.
Such partitions are called \emph{shuffled monotone sequences} (SMS).
Levcopoulos and Petersson~\cite{LevcopoulosPetersson1994}
define $\nSMS(X)$ as the minimum number of parts in such a partition.
We define its entropy counterpart as
\[
\HSMS(X)=
\min_{X=X_1\sqcup\dots\sqcup X_k}
\sum_{i=1}^k \frac{|X_i|}{n}\log\frac{n}{|X_i|}
\]
where the minimum is over all partitions into monotone subsequences.
Every SUS partition is also an SMS partition, so
$\HSMS\preceq\HSUS$ and $\nSMS\preceq\nSUS$.
Moreover, every partition into $\nSMS$ parts has entropy
at most $\log\nSMS$, so $\HSMS\preceq\nSMS$.
Our entropy-minimizing variants are therefore stronger than
their classical counterparts, and independent of the choice
of a particular partition.
We further study these relations in \cref{sect:reg_relations}
and Appendix~\ref{sect:arrows}.
Our main point is that $\reg$ dominates these strong,
historically motivated measures.

\subsection{The strength of $\reg$}\label{sect:reg_relations}

In \cref{fig:hierarchy}, $\reg$ dominates all measures except
$\Reg$, $\Hist$, and $\Loc$.
These three measures all rely on ranks within a prefix of $X$, rather than within the full input. 
Replacing them by their full-rank counterparts
$\reg$, $\hist$, and $\loc$, leaves $\reg$ as the
strongest measure.
We illustrate this by proving
$\reg\preceq\loc$, $\reg\preceq\hist$, and
$\reg\preceq\HSMS$.

\begin{lemma}\label{lemm:reg_loc}
    $\reg\preceq\loc$.
\end{lemma}
\begin{proof}
    Choosing $f(i)=i-1$ gives $
    \reg(X)\leq
    \prod_{i=2}^n\max\{1,|X[i]-X[i-1]|\}
    =\loc(X)$.
\end{proof}

\begin{lemma}\label{lemm:reg_hist}
    $\reg\preceq\hist$.
\end{lemma}
\begin{proof}
    First suppose that $X$ contains no duplicates.
    By \cref{lemm:reg_loc},
    $\reg(X^{-1})\leq\loc(X^{-1})=\hist(X)$.
Moreover, recall that $G(X)$ is the complete graph on $[n]$
with edge weights
$w(\{i,j\})=\log(2+|i-j|+|X[i]-X[j]|)$.
Since $X$ is a permutation, $G(X)$ and $G(X^{-1})$
are isomorphic: mapping each vertex $i$ to $X[i]$
exchanges indices and values, preserving every edge weight.
    Thus, by \cref{theo:reg_mst},
    $
    n+\log\reg(X)\in
    \Theta(n+\log\reg(X^{-1}))
    \subseteq O(n+\log\hist(X))$.

    If $X$ contains duplicates, recall that
    $\reg(X)\leq\reg(\ol{X})$ and
    $\hist(X)=\hist(\ol{X})$.
    Applying the previous argument to the permutation
    $\ol{X}$ proves the claim.
\end{proof}

\begin{lemma}\label{lemm:reg_hsms}
    $\reg\preceq\HSMS$.
\end{lemma}
\begin{proof}
    Let $X=X_1\sqcup\dots\sqcup X_k$ be any partition
    into monotone subsequences.
    For each $X_i=(X[a_1],\dots,X[a_t])$, where
    $a_1<\dots<a_t$, set $f(a_j)=a_{j-1}$ for $j>1$,
    and $f(a_1)=1$ if $a_1>1$.
    This defines a finger function.
    Recall the edge weights
    $w(\{i,j\})=\log(2+|i-j|+|X[i]-X[j]|)$
    from \cref{theo:reg_mst}.
    Since $X_i$ is monotone, and all values are ranks in
    $[n]$, we have
    \begin{align*}
    &\sum_{j=2}^t
    \big((a_j-a_{j-1})+
    |X[a_j]-X[a_{j-1}]|\big)\leq 2n. \text{    Thus, by concavity of the logarithm, for $t>1$,}
    \\
    &\sum_{j=2}^t w(\{a_j,a_{j-1}\})
    \leq (t-1)\log\left(2+\frac{2n}{t-1}\right) \in O\left(t\left(1+\log\frac nt\right)\right).
    \end{align*}
    The possible additional edge $\{a_1,1\}$ has
    weight $O(\log n)$.
    Since $\log n\leq t(1+\log(n/t))$, the total
    weight contributed by $X_i$ is
    $O(t(1+\log(n/t)))$, also when $t=1$.
    Summing over all subsequences gives
    $
    n+\log\reg(X)
    \in O\left(n+\sum_{i=1}^k
    |X_i|\log\frac{n}{|X_i|}\right)
=O\left(n(1+H(X_1,\dots,X_k))\right)$. 
    Taking the minimum over all monotone partitions
    shows $\reg\preceq\HSMS$.
\end{proof}

Finally, we show in Appendix~\ref{sub:regthm} that any $\reg$-adaptive algorithm is also $\Reg$-adaptive, if the input has no duplicates.:

\begin{restatable}{theorem}{regthm} \label{theo:Reg_reg}
    If $X$ is a list without duplicates, then $\log \reg \in O(n + \log \Reg)$.
\end{restatable}

\section{SproutSort in the BST model}
\label{sub:bst_def}

We study the BST model where the input  $X$ is a permutation of $[n]$ with $n$ a power of two. Note that we can transform any input $X'$ with distinct keys into $X$ in $\Reg$-adaptive time by sorting $X'$ using ~\cite{BadoiuColeDemaineIacono2007} and padding.
We show a BST-access algorithm that runs in $\reg$-adaptive time.
We extend the algorithm to allow $X$ to have duplicate keys in \cref{sub:duplicates}.

\subparagraph{The offline BST model.}
We study the cost of accessing $X$ in the offline BST model~\cite{SleatorTarjan1985,DemaineEtAl2009}.
This model maintains a \emph{single} BST containing exactly the distinct values occurring in $X$.
Each value is stored at one node; at every node, all values in the left subtree are smaller and all values in the right subtree are larger.
The algorithm knows $X$ in advance and must access $X[1],X[2],\dots,X[n]$ in order.
Every access starts with a pointer at the root and must visit the node storing the requested value.
At unit cost, the algorithm may move the pointer to its parent or either child, or rotate the current node with its parent.
The tree persists between accesses.
The total cost is $n$ plus the number of pointer moves and rotations.

\subparagraph{BST merge sort.}
A \emph{BST merge sort}~\cite{BlellochDobson2023} operates on \emph{multiple} BSTs and builds a BST that contains the input values.
A size-one input yields a one-node BST.
Otherwise, the algorithm picks $k\in[n-1]$, recursively computes BSTs $T_1$ and $T_2$ for $X[1..k]$ and $X[(k+1)..n]$, and combines them by a \emph{BST merge}:
this operation intuitively merges two trees $T_1$ and $T_2$ by selecting a connected set of nodes $S_1$ and $S_2$ (that each contain the root) and arbitrarily re-arranging their parent and child pointers such that the result is a BST that contains the elements of $T_1$ and $T_2$.
The cost of this merge is then $|S_1| + |S_2|$.
Formally, for $T_1$ and $T_2$, mark a connected set of nodes $S_1$ and $S_2$ containing the root
and the \emph{block boundaries}: nodes in $T_1$ and $T_2$ whose predecessor or successor in $X$ lies in the other $T_i$. 
The merge returns a BST $T$ on exactly the nodes of $T_1$ and $T_2$, minus duplicates, such that $S:=S_1\cup S_2$ is connected and contains the root of $T$.
The total cost is the sum of all merge costs. Blelloch and Dobson~\cite{BlellochDobson2023} show a strong link between merge sort and the offline BST model:

\begin{theorem}[Theorem~6 in~\cite{BlellochDobson2023}]
\label{theo:bst_merge_sort}
A BST merge sort that sorts a permutation $X$ with cost $f(X)$ induces an offline BST algorithm that accesses $X[1],X[2],\dots,X[n]$ with cost $O(f(X))$.
\end{theorem}

Their same proof directly applies to input lists with duplicates! Thus, to obtain good access bounds in the BST model, we may focus on designing a BST merge sort algorithm. \vspace{-0.2cm}

\subsection{SproutSort as an algorithm}
We first introduce SproutSort as a classical algorithm.
In Section~\ref{sub:bst} we change it into a BST merge sort.
Classical merge sort (\cref{algo:merge_sort}) recursively sorts the two halves of $X$ and merges them in linear time. 
On $X$, 
the merge-sort recursion tree $R$ is then a complete binary tree with $n$ leaves. 
Our core idea is to store the sorted sublists corresponding to recursive calls as a forest of \emph{sprout trees}, allowing some merges to take sublinear time.

\begin{algorithm}[H]
  \begin{algorithmic}
    \Function{MergeSort}{$X$}
        \If{$|X| = 1$} \Return $X$
        \EndIf
        \State $n \gets |X|$
        \State $X_\ell \gets \textsc{MergeSort}(X[1], \dots, X[\frac{n}{2}])$
        \State $X_r \gets \textsc{MergeSort}(X[\frac{n}{2}+1], \dots, X[n])$
        \State \Return $\operatorname{Merge}(X_\ell, X_r)$
    \EndFunction
  \end{algorithmic}
  \caption{Recursive Merge Sort } 
  \label{algo:merge_sort} 
\end{algorithm}
\vspace{-0.3cm}

\subparagraph{Sprout forests.}
Let $T$ be a complete binary tree with leaves  $1, \dots, n$, and associate the element $X[i]$ with the leaf $X[i]$.
The \emph{height} of a node $u \in T$ is the path length from $u$ to a leaf in its subtree. 
For an integer $k \ge 0$ and a subset $Y \subseteq X$, the \emph{sprout forest} $T_k[Y]$ is the subgraph of $T$ induced by all nodes $u \in T$ on the root-to-leaf path of some leaf in $Y$ that have height at most $k$.
In particular, the set of leaves of $T_k[Y]$ is precisely $Y$.
See \cref{fig:sprout_trees}.
To represent $T_k[Y]$, for a fixed pair $(k, Y)$, store a root list with all roots of $T_k[Y]$ in-order.
Store  each tree in $T_k[Y]$ as a pointer structure where vertices in $T_k[Y]$ are nodes with two pointers to their left and right children. 
Each node $u \in T_k[Y]$ stores the index of the smallest and largest leaf in its subtree,
plus a pointer to the node corresponding to $u$ in $T$.

\begin{figure}[b]
    \includegraphics[]{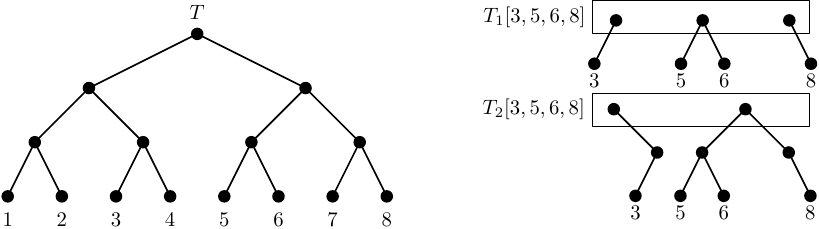}
    \caption{A complete binary tree on $X$ and a sprout forest before and after \textsc{grow}.
}
    \label{fig:sprout_trees}
\end{figure}

 Our sprout trees support two operations.
 \textsc{Grow}$(T_k[Y])$ for $k \geq 0$ and $Y \subseteq X$ computes $T_{k+1}[Y]$.
 \textsc{Merge}$(T_k[Y_1], T_k[Y_2])$ for $k \geq 0$ and $Y_1, Y_2 \subseteq X$ computes $T_k[Y_1 \cup Y_2]$ (see Figure~\ref{fig:grow_and_merge}). Both operations consume their input forests.
We analyse sequences of these operations with the accounting method~\cite{CLRS2009}.
We assign one coin to each node which can pay for $O(1)$ work.
The amortized running time of \textsc{Merge} is $O(m)$ where $m$ is the number of trees in the output.

\begin{figure}[b]
    \includegraphics[]{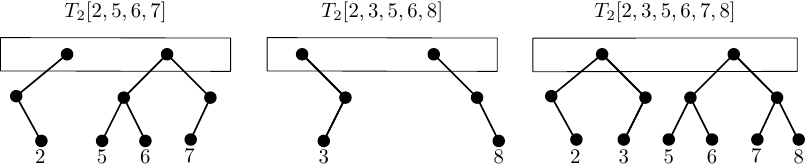}
    \caption{Two sprout forests and their merge.}
    \label{fig:grow_and_merge}
\end{figure}

\begin{theorem}\label{theo:grow_merge}
Consider any sequence of \textsc{Grow} and \textsc{Merge} operations, with $k \leq \log n$, starting with a collection of sprout forests that have $n$ nodes in total. Let $m_i$ be the size of the root-list of the forest produced by operation $i$. The total running time is $
    O\left(n+\sum_i m_i\right)$.
Moreover, \textsc{Merge} never increases the number of nodes.
\end{theorem}

\begin{proof}
Initially, we create one coin for each existing node.
Whenever we remove a node, we spend its coin to do $O(1)$ work.
Additionally, the $i$-th operation can take $O(m_i)$ time.

A \textsc{Grow}($T_k(Y)$) operation scans all roots in  $T_k(Y)$.
For each root $r$ it encounters, it either creates a new root with $r$ as its singular child or a new root with two children, one of which is $r$.
Thus, \textsc{Grow} inspects at most $3 m_i$ nodes, and, since the new root list follows the in-order traversal of the old root list, \textsc{Grow} takes $O(m_i)$ time.
\textsc{Grow} also creates $m_i$ coins.

For $\textsc{Merge}(T_k(Y_1), T_k(Y_2))$, observe that each root in $T_k(Y)$ is at least a root in either $T_k(Y_1)$ or $T_k(Y_2)$, and at most a root in both.
Thus, the total number of trees in $T_k(Y_1)$ and $T_k(Y_2)$ is at most $2 m_i$, and since root-lists of $T_k(Y_1)$ and $T_k(Y_2)$ are in-order, we can identify these roots in $O(m_i)$ time. 
Having identified colliding roots, we invoke our \textsc{TreeMerge} procedure on all trees that share a root (see Algorithm~\ref{algo:trie_merge}).
Whenever this recursive merge replaces two nodes by one, we transfer one node's coins to the replacement and spend the other's on the $O(1)$ work needed for this merge. \textsc{TreeMerge} does not create any coins.

In total, we create $O(n + \sum_i m_i)$ coins and spend $O(m_i)$ time beyond the coin budget.
\end{proof}

\begin{algorithm}[t]
  \begin{algorithmic}
    \Function{TreeMerge}{$T_1, T_2$}
        \If{$T_1$ is nil}
            \Return $T_2$
        \EndIf
        \If{$T_2$ is nil}
            \Return $T_1$
        \EndIf
        \State \Return a new node with children $\textsc{TreeMerge}(T_1.\ell, T_2.\ell)$ and $\textsc{TreeMerge}(T_1.r, T_2.r)$
    \EndFunction
  \end{algorithmic}
  \caption{Recursively merging sprout trees.} \label{algo:trie_merge}
\end{algorithm}

\vspace{-1.5em}
\subparagraph{SproutSort}
We now design an algorithm \textsc{SproutSort} that uses \textsc{Grow} and \textsc{Merge} to compute $T_k[X]$ where $|X|=2^k$.
Similar to recursive merge sort, the algorithm splits $X=X_\ell\sqcup X_r$ into equal halves and recursively computes $T_{k-1}[X_\ell]$ and $T_{k-1}[X_r]$.
It then returns $T_k[X]=\textsc{Grow}(\textsc{Merge}(T_{k-1}[X_\ell],T_{k-1}[X_r]))$.
See \cref{algo:sprout_sort}.

\begin{definition}
Recall that a \emph{finger function} is a map $f:[n]\setminus\{1\}\to[n]$ satisfying $f(i)<i$.
For $i\geq 2$, we define the graph distances:
\begin{itemize}
\item $\texttt{dist}_R(i,f(i))$, between leaves $i$ and $f(i)$ in the complete binary tree $R$ on $[n]$.
\item $\texttt{dist}_T(i,f(i))$, between leaves $X[i]$ and $X[f(i)]$ in $T$.
\end{itemize}
\end{definition}

\noindent
For example, for $n\geq4$, $\texttt{dist}_R(\frac{n}{2}, \frac{n}{2}+1) = 2\log n$ and $\texttt{dist}_R(\frac{n}{2} + 1, \frac{n}{2}+2) = 2$. 

{
\begin{algorithm}[t]
\begin{algorithmic}
\Function{SproutSort}{$X$}
\If{$|X|=1$}
\Return $T_0[X]$
\EndIf
\State $m\gets |X|,\ k\gets\log_2 m$
\State $X_\ell\gets X[1],\dots,X[m/2]$, $X_r\gets X[m/2+1],\dots,X[m]$
\State $T_{k-1}[X_\ell]\gets\textsc{SproutSort}(X_\ell)$
\State $T_{k-1}[X_r]\gets\textsc{SproutSort}(X_r)$
\State \Return $\textsc{Grow}(\textsc{Merge}(T_{k-1}[X_\ell],T_{k-1}[X_r]))$
\EndFunction
\end{algorithmic}
\caption{Recursive SproutSort.}\label{algo:sprout_sort}
\end{algorithm}
}

\begin{lemma}\label{lemm:sprout_sort}
Let $n\geq 2$ and let $f$ be any finger function.
The running time of \textsc{SproutSort} is
$O\bigl(\sum_{i=2}^{n}(1+\texttt{dist}_R(i,f(i))+\texttt{dist}_T(i,f(i)))\bigr)$.
\end{lemma}
\begin{proof}
Let $g$ be the total number of nodes created by \textsc{Grow}.
Observe that \textsc{Grow} creates one node per tree in its output forest, \textsc{Merge} does not increase the number of nodes,
and both \textsc{Grow} and \textsc{Merge} reduce the size of the root list by at most a factor $2$.
Hence, an invocation of $\textsc{Grow}(\textsc{Merge}(\dots))$ in \textsc{SproutSort} that creates $g_i$ nodes,
involves root lists of size at most $4 g_i$.
Thus, by \cref{theo:grow_merge}, the total running time is $O(n + \sum_i g_i) = O(n + g)$.

Whenever \textsc{Grow} creates a node $u$, we denote by $\subtree(u)$ the resulting sprout tree rooted at $u$. We charge the creation of $u$ to the smallest $i\geq 2$ such that $X[i]\in\subtree(u)$ but $X[f(i)]\notin\subtree(u)$.
If no such $i$ exists, we charge $u$ to item $1$.
This happens only if $X[1]\in\subtree(u)$, which itself happens for at most one node $u$ per level of recursive depth in our algorithm, for a total charge of $O(\log n)$ to item $1$.

For all other $i \geq 2$, we bound the charge they receive.
At every recursion height $k\geq \texttt{dist}_R(i,f(i))/2$, items $i$ and $f(i)$ belong to the same recursive input $X'$.
Thus, they belong to the same sprout forest $T_k(X')$. 
If, furthermore $k\geq  \texttt{dist}_T(i,f(i))/2$ then they also lie in the same sprout tree of $T_k(X')$. 
Item $i$ participates in only one \textsc{Grow} operation per recursive height, and stops being charged once items $i$ and $f(i)$ share a sprout tree.
Item $i$ therefore receives at most  $O(1 + \texttt{dist}_R(i,f(i)) + \texttt{dist}_T(i,f(i)))$ charges. 
Consequently, $g \in O\bigl(\log n+\sum_{i=2}^{n}(1+\texttt{dist}_R(i,f(i))+\texttt{dist}_T(i,f(i)))\bigr)$.
The sum is at least $n-1$, so it absorbs the $O(n)$ initialization cost and the $O(\log n)$ charges to item $1$.
\end{proof}

\begin{lemma}\label{lemm:complete_distance}
Let $X$ be a permutation of $[n]$ and let $f$ be any finger function.
Then
\scalebox{0.92}[1]{$\sum\limits_{i=2}^{n}\texttt{dist}_R(i,f(i))\in O(n+\sum\limits_{i=2}^{n}\log(i-f(i)))$ and $\sum\limits_{i=2}^{n}\texttt{dist}_T(i,f(i))\in O(n+\sum\limits_{i=2}^{n}\log|X[i]-X[f(i)]|)$.}
\end{lemma}
\vspace{-0.3cm}
\begin{proof}
We prove the bound for $\texttt{dist}_T$, which is the more difficult case because $X[f(i)]$ may lie on either side of $X[i]$.
Construct $T$ from the leaves upwards.
In round $k$, the leaves are partitioned into groups of $2^k$ consecutive values.
If two leaves start belonging to the same group in round $t$, their lowest common ancestor is $t$ edges above them.
Thus, $\texttt{dist}_T(i,f(i))$ is twice the number of rounds in which $X[i]$ and $X[f(i)]$ belong to different groups.

Write $\delta_i=|X[i]-X[f(i)]|$. For a fixed $i$, there are at most $1+2\log\delta_i$ rounds with $2^k\leq\delta_i^2$. We bound the remaining rounds by counting how often they occur across all indices.

Fix a round $k$ with $2^k>\delta_i^2$, and suppose $X[i]$ and $X[f(i)]$ belong to different groups. Let $[a,b]$ be the group containing $X[i]$. If $X[f(i)]<a$, then $X[i]-a+1\leq\delta_i$; if $X[f(i)]>b$, then $b-X[i]+1\leq\delta_i$. Hence $X[i]$ is among the first or last $\delta_i<2^{k/2}$ values of its group. There are $n/2^k$ groups, and $X$ is a permutation, so this happens for at most $2n/2^{k/2}$ indices in round $k$. Summing over all rounds gives at most $2n\sum_{k=0}^{h-1}2^{-k/2}=O(n)$ such occurrences.

The first kind of round contributes $O\bigl(n+\sum_{i=2}^{n}\log\delta_i\bigr)$ occurrences, and the second contributes $O(n)$. Each occurrence contributes two to the distance, proving the lemma.
\end{proof}

\begin{theorem}
    If $X$ is a permutation and $T$ is a complete binary tree,
    then the running time of \textsc{SproutSort} is $O(n + \log \reg)$.
\end{theorem}
\begin{proof}
    Let $f$ be an arbitrary finger function.
    By \cref{lemm:sprout_sort}, the running time is $O\bigl(\sum_{i=2}^{n}(1+\texttt{dist}_R(i,f(i))+\texttt{dist}_T(i,f(i)))\bigr)$.
    By Lemma~\ref{lemm:complete_distance}, the running time is therefore upper-bounded by 
    $O\bigl(n + \sum_{i=2}^{n}\log(i-f(i)) +  \sum_{i=2}^{n}\log |X[i]-X[f(i)]| \bigr)$.
    By the definition of $\reg$, there exists some finger function $f$ such that $O(n + \log \reg) = O\bigl(n + \sum_{i=2}^{n}\log(i-f(i)) +  \sum_{i=2}^{n}\log |X[i]-X[f(i)]| \bigr)$ and so the theorem follows.
\end{proof}

\subsection{Running \textsc{SproutSort} in the BST Model} \label{sub:bst}

What remains is to implement \textsc{SproutSort} as a BST merge sort algorithm, which then implies an algorithm for BST access.
That is, we must implement our recursive algorithm so that the output of each recursive call is a BST whose nodes are exactly the input of the call.
The input is a list of singleton BSTs whose left-to-right order is $X$.
Each recursive call combines its two output BSTs using the \textsc{Merge} operator from the beginning of the section.

Consider a recursive call of \textsc{SproutSort} at recursion height $k$ on a contiguous sublist $X'$ of $X$.
To represent a sprout tree $T' \in T_k[X']$ as a BST, first contract all its maximal unary paths, yielding a tree $T''$.
Assign each internal node $u$ of $T''$ the key of the rightmost leaf in its left subtree, unless this subtree consists of a single leaf.
Similarly, assign it the key of the leftmost leaf in its right subtree, unless this subtree consists of a single leaf.
Split every node that is assigned two keys, putting the larger key above the smaller.
Remove the original leaves and any internal nodes that were not assigned a key, yielding a BST $T'''$ containing all keys except the leftmost and rightmost.
Finally, make the rightmost key $v$ a new root, the leftmost key $u$ the left child of $v$, and $T'''$ the right subtree of $u$.
A one-leaf sprout tree becomes a singleton BST; see \cref{fig:sprout_tree_bst}.
To represent a forest $T_k[X']$, simply represent each sprout tree in this way and connect the roots in a linked list via their right child pointers.
This yields a BST, as the root list of a sprout forest contains the items in increasing order.

\begin{figure}[t]
    \includegraphics[width=\linewidth]{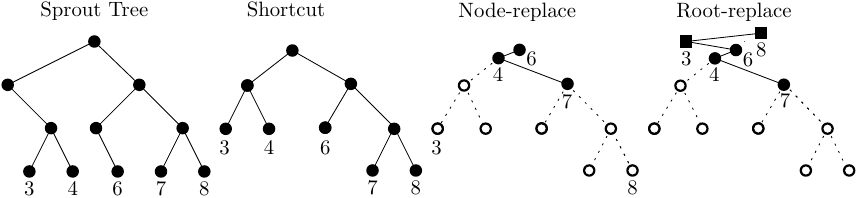}
    \caption{Converting a sprout tree to a BST in three steps. White nodes are implicit; each root-to-leaf path has a corresponding BST traversal.}
    \label{fig:sprout_tree_bst}
\end{figure}

\subparagraph{Our operations.}
The \textsc{Grow} operation traverses the root list of a sprout forest and may attach pairs of adjacent trees to a new root.
In the BST representation, this corresponds to traversing the root list through right-child pointers and performing $O(1)$ rotations per pair.
The number pointer traversals and rotations performed is bounded by the root list size.

\begin{figure}[b]
    \includegraphics[width=\textwidth]{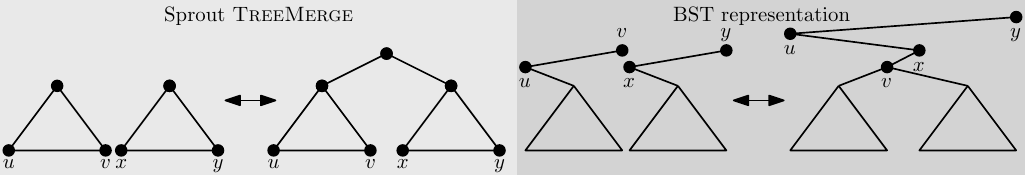}
    \caption{The primitives used in Sprout \textsc{Tree-Merge} (left) and their BST representation (right).}
    \label{fig:sprout_tree_bst_merge}
\end{figure}

The \textsc{Merge} operation, when given two disjoint BSTs, traverses their root lists and calls \textsc{TreeMerge} (\cref{algo:trie_merge}) on sprout trees that share the same root.
Each recursive call to \textsc{TreeMerge} separates the root of each sprout subtree $T_i$ from its children $T_i.\ell,T_i.r$, recursively merges the left and right children, and then re-attaches the results to a new root.
In the BST representation, separating its root from its two children involves up to four nodes near the root.
Re-attaching to a new root is the reverse (see \cref{fig:sprout_tree_bst_merge}).
For each separation operation, we mark the roots of the resulting binary search trees, plus their left children (nodes $u,v,x,y$ in \cref{fig:sprout_tree_bst_merge}).
At the start of \textsc{Merge}, we also mark all nodes on the right spine of each input BST (i.e., the root list plus their left children).
With these markings, the marked nodes of each input BST form a connected subtree containing its root, and unmarked subtrees remain unchanged.
Moreover, every block boundary between the two input BSTs is exposed during \textsc{TreeMerge} and therefore marked.
Thus, $\textsc{Grow}(\textsc{Merge}(\dots))$ in \textsc{SproutSort} is a BST merge.
The number of marked nodes is bounded by the number of root-list entries traversed and nodes processed by \textsc{TreeMerge}.
By \cref{theo:grow_merge} and our analysis of \textsc{SproutSort}, the total BST merge cost is $O(n+\log\reg(X))$, which shows:

\begin{theorem}
    \textsc{SproutSort} is a BST merge sort~\cite{BlellochDobson2023}.
    Thus, it yields an offline BST algorithm for accessing a permutation $X$
    in $O(n+\log\reg(X))$ node traversals and rotations.
\end{theorem}


\subsection{Duplicate Items}
\label{sub:duplicates}

So far, we assumed that $X$ has no duplicates.
We now show that \textsc{SproutSort} can be adapted to the case where $X$ contains duplicate items.
In the offline BST model for accessing elements, this corresponds to having an access sequence $X$ that contains the same access multiple times, where the BST stores the set $Y$ of distinct keys in $X$.
Although input elements can only be compared, we can transform any input $X'$ into a list $X \in [U]^n$, where $U$ is the number of distinct keys in $X$ (via~\cite{BadoiuColeDemaineIacono2007}) in $O(n + \log \reg)$ time. 
Furthermore, we can assume that we know for each $x \in [U]$ how many times $x$ occurs in $X$.

\begin{definition}
    \label{def:tree_T}
    We denote by $T$ a specific binary tree whose leaves are $1,\dots,U$ in-order and have the same depth. $T$ could have many shapes, as we allow nodes with only one child. We construct $T$ cleverly later, and associate each $X[i]$ with the leaf that stores this value. 
\end{definition}

Given a specifically-shaped, currently-unspecified $T$, we run \textsc{SproutSort} as before where we use $T$ as the skeleton for the \textbf{Sprout Trees}, as opposed to the complete binary tree. 
Each recursive call on $X[i..j]$ returns a sprout forest of the distinct values in $X[i..j]$.
Separate recursive calls store separate copies of their forests, and different forests may contain leaves with the same value.
Each leaf with value $x\in[U]$ stores a linked list of the occurrences of $x$ in the corresponding recursive input $X[i..j]$.
When \textsc{TreeMerge} encounters two copies of the same leaf, it replaces them by one and concatenates their occurrence lists in $O(1)$ time.

\subparagraph{Recreating \cref{theo:grow_merge} and \cref{lemm:sprout_sort}}
The amortisation proof of \cref{theo:grow_merge} applies as-is: identifying two copies of a leaf spends the deleted copy's coin on the $O(1)$ work required to identify the duplicate, and the linked list merge. 
We do have to extend \cref{lemm:sprout_sort}:

\begin{lemma}\label{lemm:sprout_sort_general}Let $T$ be the tree from Definition~\ref{def:tree_T}, $f$ be any finger function and $n \geq 2$.
\textsc{SproutSort} on  input $X \in [U]^n$ runs in time
$O\bigl(n+\sum_{i=2}^{n}(\log(i-f(i))+\texttt{dist}_T(i,f(i)))\bigr)$.
\end{lemma}

\begin{proof}
Let $g$ be the total number of nodes created by \textsc{Grow}.
Let $g_j$ be the number of nodes created by \textsc{Grow} in the $j$-th recursive call; then $g:=\sum_j g_j$.
If \textsc{Grow} creates nodes, it creates one per tree in its output forest.
Otherwise, its input forest already contains the root of $T$, and it returns this single tree unchanged.
Thus, its output root list has size at most $g_j+1$.
Since \textsc{Grow} combines at most two input trees into each output tree, the preceding \textsc{Merge} has an output root list of size at most $2(g_j+1)$.
By \cref{theo:grow_merge}, the total running time is
$O(n+\sum_j(g_j+1))=O(n+g)$,
since there are $n-1$ non-base recursive calls.

Whenever \textsc{Grow} creates a node $u$ in a recursive call on $X[a..b]$, let
$C_u:=\{i\in[a,b]\mid X[i]\in\subtree(u)\}$ be all input keys in the current recursive call that lie below $u$.
We charge the creation of $u$ to the smallest $i\geq2$ such that $i\in C_u$ but $f(i)\notin C_u$.
If no such $i$ exists, the smallest index in $C_u$ must be $1$, since $f(i)<i$.
We then charge $u$ to item $1$.
At each recursion height, only one recursive input contains item $1$, and only one newly created sprout tree contains its leaf.
Thus, item $1$ receives $O(\log n)$ charges.

For all other $i\geq2$, at every recursion height $k\geq\texttt{dist}_R(i,f(i))/2$, items $i$ and $f(i)$ belong to the same recursive input.
If, furthermore, $k\geq\texttt{dist}_T(i,f(i))/2$, their leaves belong to the same sprout tree.
Item $i$ therefore receives at most $O(1+\texttt{dist}_R(i,f(i))+\texttt{dist}_T(i,f(i)))$ charges.
This also holds when $X[i]=X[f(i)]$, in which case $\texttt{dist}_T(i,f(i))=0$.
Every node created by \textsc{Grow} charges exactly one item, so the total number of charges is $\sum_j g_j$.
Summing the bounds above, including the $O(\log n)$ charges to item $1$, gives
$\sum_j g_j\in O\bigl(\log n+\sum_{i=2}^{n}(1+\texttt{dist}_R(i,f(i))+\texttt{dist}_T(i,f(i)))\bigr)$.
Since we could upper bound our running time by $O(n+\sum_j g_j)$, we can now upper bound our running time by 
$O\bigl(n+\sum_{i=2}^{n}(\texttt{dist}_R(i,f(i))+\texttt{dist}_T(i,f(i)))\bigr)$.
What remains is to replace all $\texttt{dist}_R$ terms.
For this, we can directly invoke \cref{lemm:complete_distance}, as $R$ is completely unaffected by $X$ containing duplicates.
\end{proof}

\subparagraph{Choosing the tree of Definition~\ref{def:tree_T}.}
By \cref{lemm:sprout_sort_general}, what remains is to bound the total distance between leaves representing $X[i]$ and $X[f(i)]$ in our custom-made tree $T$. 
We choose $T$ using the number of occurrences of each value via some weight function:

\begin{lemma}\label{lemm:weighted_tree}
Let $w:[U]\to[n]$ be a weight function with $\sum_{x=1}^{U}w(x)=n$.
There is a binary tree $T$, possibly containing nodes with one child,
whose leaves are $1,\dots,U$ in-order and have the same depth, such that,
putting
$E(x):=\max_{y\in[U]}\{\texttt{dist}_T(x,y)-8\log\max\{1,|x-y|\}\}$,
we have $\sum_{x=1}^{U}w(x)E(x)\in O(n)$.
$T$ has $O(U)$ nodes and can be constructed in $O(U)$ time.
\end{lemma}
\begin{proof}
For an integer interval $[a,b]\subseteq[U]$, we construct a tree
$T'_{a,b}$ whose leaves are $a,\dots,b$.
If $a=b$, then $T'_{a,b}$ consists of a single leaf labelled $a$.
Otherwise, choose an index $r\in\{a,\dots,b-1\}$ by the rule below.
Construct $T'_{a,r}$ and $T'_{r+1,b}$ recursively, and make their roots
the left and right children of a new root.
The resulting tree is $T'_{a,b}$.
Finally, put $T'=T'_{1,U}$.

We choose $r$ as follows.
Put $J=[a,b]$, $m=b-a+1$, and $W=\sum_{x\in J}w(x)$.
Require $m/3\leq r-a+1\leq2m/3$, so both child intervals contain
at most $2m/3$ keys.
Put $q=\lceil\sqrt m\rceil$.
Starting with the smallest permitted index $r_0$, consider
$r_0,r_0+2q,r_0+4q,\dots$, stopping when the index is no longer permitted.
For each candidate $r$, put $B_r=J\cap[r-q+1,r+q]$.
These are disjoint, and there are $\Theta(\sqrt m)$ candidates with total weight at most $W$, so one has weight $O(W/\sqrt m)$.
Choose an $r$ minimizing this weight, and call the keys in $B_r$ \emph{bad for $T'_{a,b}$}.

We now turn $T'$ into $T$ by lengthening its edges.
For each node $u$ of $T'$, let $m_u$ be its number of descendant leaves
and assign it the target height $h(u)=\lfloor2\log m_u\rfloor$.
Replace every parent--child edge $pu$ by a path of $h(p)-h(u)$ edges,
inserting nodes with one child.
This length is positive because $m_u\leq2m_p/3$.
Along any path from an original node $u$ to a descendant leaf,
the new edge lengths telescope to $h(u)$, since leaves have target height zero.
Thus all leaves have the same depth. 
What remains is to bound $\sum_x w(x)E(x)$ in $T$.

For distinct leaves $x,y$, their lowest common ancestor is the same node $u$ in both trees, and
$\texttt{dist}_T(x,y)=2h(u)$.
Suppose leaves in the subtree of $u$ form an interval of size $m$,
partitioned at the chosen index $r$.
The keys $x,y$ lie in different child intervals.
If $x\notin B_r$, then $|x-y|\geq q\geq\sqrt m$, so
$\texttt{dist}_T(x,y)=2h(u)\leq4\log m\leq8\log|x-y|$.
This pair contributes no positive excess to $E(x)$.
If $x\in B_r$, its excess is at most $4\log m$.
Since $y=x$ contributes zero, $E(x)$ is at most the sum of $4\log m$,
over all nodes at which $x$ is bad.

Consider a node whose descendant interval has size $m$ and weight $W$.
Its bad keys have weight $O(W/\sqrt m)$.
The contribution of these bad keys to the weighted sum is therefore $O(W\log m/\sqrt m)$.
Charge each key $x$ (good or bad) in this interval $O(w(x)\log m/\sqrt m)$.
The nodes whose intervals contain $x$ are the ancestors of $x$ in $T'$,
and their interval sizes shrink by a factor at most $2/3$.
For each $k$, there are $O(1)$ nodes with $2^k\leq m<2^{k+1}$.
Hence the total charge to $x$ is
$O(w(x)\sum_{k\geq1}(k+1)2^{-k/2})=O(w(x))$.
Summing over $x$ gives $O(n)$.

To get the running time, compute for $t \in [U]$ the prefix sums $P(t)$ where $P(0)=0$ and
$P(t)=\sum_{x=1}^{t}w(x)$ in $O(U)$ time.
The weight of each candidate's bad interval $[L,R]$ is then
$P(R)-P(L-1)$, which takes $O(1)$ time to compute. 
Each leaf $T'_{a,a}$ is created in constant time.
At an internal node with $m$ descendant leaves, choosing $r$ takes
$O(\sqrt m)$ time, and lengthening its outgoing edges inserts
$O(\log m)$ nodes.
There are $O(U/2^k)$ nodes of $T'$ with $2^k\leq m<2^{k+1}$:
each has at least $2^k$ descendant leaves, and each leaf lies below
only $O(1)$ such nodes.
Summing over $k$ gives $O(U)$ construction time and nodes.
\end{proof}

\begin{theorem}\label{theo:sprout_sort_general}
Let $X\in[U]^n$, where we know for each element $x\in X$ its multiplicity in $X$.
With $T$ constructed by \cref{lemm:weighted_tree}, \textsc{SproutSort} sorts $X$ in $O(n+\log\reg(X))$ time.
\end{theorem}

\clearpage
\begin{proof}
For each $x\in[U]$, let $w(x)$ be its multiplicity in $X$, and construct $T$ by \cref{lemm:weighted_tree}; this takes $O(n)$ time.
Let $f$ be any finger function and put $\delta_i:=\max\{1,|X[i]-X[f(i)]|\}$.
Since all leaves of $T$ have the same depth,
$\texttt{dist}_T(i,f(i))=2\height(\LCA_T(X[i],X[f(i)]))\leq8\log\delta_i+2E(X[i])$.
This includes equal-value pairs, whose tree distance is zero.
Furthermore, $\sum_{i=2}^{n}E(X[i])\leq\sum_{x=1}^{U}w(x)E(x)\in O(n)$.
By \cref{lemm:sprout_sort_general}, the running time is therefore
$O\bigl(n+\sum_{i=2}^{n}\log(i-f(i))+\sum_{i=2}^{n}\log\delta_i\bigr)$.
By the definition of $\reg$, there is a finger function $f$ for which these two sums equal $\log\reg(X)$.
Traversing the output leaves and their occurrence lists takes $O(n)$ additional time and yields a sorted ordering of $X$.
\end{proof}

Finally, we transform the sorting algorithm into a BST access algorithm in the same manner as in Section~\ref{sub:bst}.
We use the BST representation and markings of \cref{sub:bst}, with one node per distinct key in each recursive output.
When \textsc{TreeMerge} identifies two copies of a leaf, both corresponding BST nodes are marked before they are identified.
Thus, every key present in both input trees is marked in both.
Each input BST contains one node per distinct key. However, the same key may occur in both input trees, so a merge may encounter two nodes storing that key. We mark both nodes and combine them into a single node in the output BST. For example, if both input trees contain a node storing $3$, the output tree contains only one node storing $3$. Both input nodes count towards the merge cost. Unmarked subtrees retain their nodes and children. Combining equal-key nodes does not remove occurrences from the access sequence: the resulting BST still accesses a key once for each occurrence in $X$.

The extension preserves the access bound of \cref{theo:bst_merge_sort}.
In Blelloch and Dobson's construction~\cite{BlellochDobson2023}, marked nodes with the same key produce identical geometric points.
Unmarked subtrees contain no shared keys, so their structure and relative priorities remain unchanged.
The resulting set is therefore still arborally satisfied, with $O(n)$ points plus $O(1)$ per marked input node, giving the same asymptotic BST access cost.

\begin{theorem}
    \textsc{SproutSort} yields a deterministic offline BST algorithm for accessing $X$ with cost $O(n+\log\reg(X))$, also when $X$ contains duplicates.
\end{theorem}
\begin{proof}
We use the BST representation and markings of \cref{sub:bst}, identifying marked copies of the same key as described above.
Each recursive call performs a BST merge with the duplicate extension described above.
The number of marked input nodes is bounded by the work performed by \textsc{Grow} and \textsc{Merge}.
By \cref{theo:sprout_sort_general}, the total merge cost is therefore $O(n+\log\reg(X))$.
The duplicate extension of \cref{theo:bst_merge_sort} then gives an offline BST algorithm with the claimed access cost.
\end{proof}

\newpage
\appendix

\newpage
\section{SkeletonSort with no duplicates}
Let $X$ be an input list of arbitrary comparable elements that contains no duplicates. 
We present a simple, randomised algorithm that sorts $X$ in expected $O(n + \log \reg)$ time, without any prior knowledge of the ranks of elements in $X$. 
Our algorithm is based on Quicksort, which picks a uniformly random pivot $p$ and partitions $X$ according to whether elements are smaller or larger than $p$ (recursing on each part).
We design a tree data structure that supports both pivot sampling and partitioning, in expected $O(n + \log \reg)$ total time. 
By Theorem~\ref{theo:Reg_reg}, the algorithm is not only $\reg$-adaptive but also $\Reg$-adaptive.

\begin{algorithm}[H]
  \begin{algorithmic}
    \Function{Quicksort}{$X$}
        \If{$X$ is empty}
            \State \Return an empty list
        \EndIf
        \State $p \gets \operatorname{random-item}(X)$
        \State $L_\ell \gets \textsc{Quicksort}(\{x\in X:x<p\})$
        \State $L_r \gets \textsc{Quicksort}(\{x\in X:x>p\})$
        \State \Return the concatenation of $L_\ell$, the singleton list $[p]$, and $L_r$
    \EndFunction
  \end{algorithmic}
  \caption{Quicksort.}\label{algo:quicksort}
\end{algorithm}

\subparagraph{Skeleton Trees.}
We again assume that $n=2^k$ is a power of two (we may always pad the input with distinct increasing sentinel items declared larger than every key in $X$, without it asymptotically changing $O(n + \log \reg)$).
Let $R$ be the complete binary tree whose leaves are the indices $1,\dots,n$ in-order.
For a subset of indices $I\subseteq[n]$, we define the \emph{$I$-skeleton of $R$}, denoted $R_{I}$, as the tree obtained by keeping the leaves in $I$ and their ancestors, and shortcutting every node with only one child.
For $I=\emptyset$, put $R_{I}=\mathrm{nil}$.
We call such trees \emph{skeleton trees}, and they are the basis of our \textsc{SkeletonSort} algorithm.
Each internal node of a skeleton tree has exactly two children, and each node $u$ corresponds to an original node of $R$; the \emph{raw height} of $u$ is the number of edges from that original node to any of its descendant leaves in $R$.
We store each skeleton tree as a pointer structure with child and parent pointers.
A leaf labelled $i$ stores a pointer to the input item $X[i]$.
Each node stores its number of descendant leaves, together with the minimum and maximum of the input values $X[i]$ represented by those leaves.
For each internal node, these fields can be computed from the two children in constant time.

\begin{figure}[b]
    \includegraphics[]{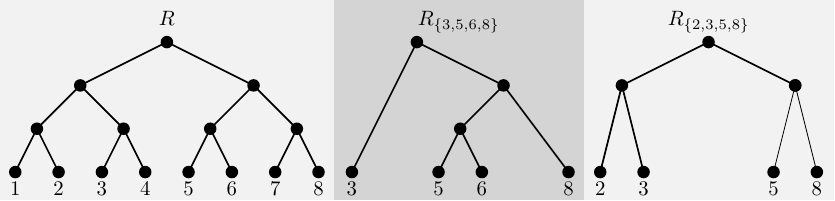}
    \caption{A complete binary tree and two skeleton trees.}
    \label{fig:skeleton_trees}
\end{figure}

\subparagraph{Random sampling.}
To simulate Quicksort, we need for any $I \subset [n]$ to be able to generate a uniformly random index $i \in I$. 
We show how to do this given $R_{I}$:
the root of this tree has stored $|I|$. Choose a uniformly random integer $j\in[|I|]$ and find the $j$-th leaf in-order, using the leaf counts to choose the next child.
This procedure \textsc{Random-Leaf} returns a pointer to an item sampled uniformly from $I$ and takes time proportional to the number of nodes on its search path. We allow constant-time arithmetic on $O(\log n)$-bit counters and generation of a uniformly random integer in a specified range of size at most $n$.

\subparagraph{Splitting along a pivot.}
Given some random index $i \in I$, this yields some pivot value $p = X[i]$. Quicksort then splits the input along this pivot.  
We simulate this split by defining $I^{<p}=\{i\in I:X[i]<p\}$ and $I^{=p}$ and $I^{>p}$ analogously.
The operation \textsc{Skeleton-Split}$(R_{I},p)$ (Algorithm~\ref{algo:skeleton_split}) consumes its input and returns the three disjoint skeletons
$R_{I^{<p}}$, $R_{I^{=p}}$, and $R_{I^{>p}}$.
In this code, whenever \textsc{Create-Node} creates a node, it updates the stored fields and the parent pointers of its children in constant time. Every returned non-nil tree has its root's parent pointer reset to nil.
When both children are non-nil, the new node corresponds to the same original node of $R$ as the node currently being partitioned.

\begin{algorithm}[H]
  \begin{algorithmic}
    \Function{Create-Node}{$S_\ell,S_r$}
        \If{both $S_\ell$ and $S_r$ are nil}
            \State \Return nil
        \EndIf
        \If{exactly one of $S_\ell,S_r$ is nil}
            \State \Return the non-nil tree
        \EndIf
        \State $u\gets$ a new node with children $S_\ell$ and $S_r$
        \State Update the stored fields of $u$ and the parent pointers of its children
        \State \Return $u$
    \EndFunction
    \vspace{0.5em}
    \Function{Skeleton-Split}{$S,p$}
        \If{$S$ is nil}
            \State \Return nil, nil, nil
        \EndIf
        \If{all leaves $i$ of $S$ satisfy $X[i]=p$}
            \State \Return nil, $S$, nil
        \EndIf
        \If{all leaves $i$ of $S$ satisfy $X[i]<p$}
            \State \Return $S$, nil, nil
        \EndIf
        \If{all leaves $i$ of $S$ satisfy $X[i]>p$}
            \State \Return nil, nil, $S$
        \EndIf
        \State $S_\ell^{<},S_\ell^{=},S_\ell^{>} \gets \textsc{Skeleton-Split}(S.\ell,p)$
        \State $S_r^{<},S_r^{=},S_r^{>} \gets \textsc{Skeleton-Split}(S.r,p)$
        \State $S^{<}\gets\textsc{Create-Node}(S_\ell^{<},S_r^{<})$
        \State $S^{=}\gets\textsc{Create-Node}(S_\ell^{=},S_r^{=})$
        \State $S^{>}\gets\textsc{Create-Node}(S_\ell^{>},S_r^{>})$
        \State Discard the old node $S$
        \State \Return $S^{<},S^{=},S^{>}$
    \EndFunction
  \end{algorithmic}
  \caption{Partitioning skeleton trees.}\label{algo:skeleton_split}
\end{algorithm}

\subsection{SkeletonSort.}
Given the above two primitives, we can simulate \emph{Quicksort} without using linear time per recursive call: 
We construct $R_{[n]}=R$ and its stored fields in $O(n)$ time, then run \textsc{SkeletonSort} on its root (\Cref{algo:skeleton_quicksort}).
Since the values are distinct, $S^{=}$ consists of the single leaf storing the pivot. Output lists have head and tail pointers, so each concatenation takes $O(1)$ time.

\begin{algorithm}[h]
  \begin{algorithmic}
    \Function{SkeletonSort}{$S$}
        \If{$S$ is nil}
            \State \Return an empty list
        \EndIf
        \State $p\gets$ the item returned by $\textsc{Random-Leaf}(S)$
        \State $S^{<},S^{=},S^{>} \gets \textsc{Skeleton-Split}(S,p)$
        \State $L_\ell \gets \textsc{SkeletonSort}(S^{<})$
        \State $L_r \gets \textsc{SkeletonSort}(S^{>})$
        \State \Return the concatenation of $L_\ell$, the singleton list $[p]$, and $L_r$
    \EndFunction
  \end{algorithmic}
  \caption{SkeletonSort.}\label{algo:skeleton_quicksort}
\end{algorithm}

To analyse this algorithm,  fix a random seed and form a tree $T_0$ as follows.
Create one node for every non-empty recursive invocation of \textsc{SkeletonSort}.
Its children are the nodes for its non-empty recursive invocations, and the node is labeled by its pivot.
Then, $T_0$ is a binary search tree that contains the values of $X$ in sorted order.

For each node $v$ of $T_0$, let $h(v)$ be the height of the subtree of $v$, i.e. the maximum number of edges from $v$ to a descendant leaf.
The node $v$ is labeled by some pivot $x$.
Remove this label and instead attach a path of length $h(v)$ to $v$ whose leaf is labeled by $x$.
Finally, replace every parent--child edge $vu$ by a path of $h(v)-h(u)$ edges, inserting nodes with one child, and call the resulting tree $T$.
Every path from an original node $v$ to a descendant leaf now has length $h(v)$, so all leaves of $T$ have the same depth.
The resulting tree $T$ is used only in the analysis; the algorithm does not construct it.

\begin{tcolorbox}[
    colback=gray!8,
    colframe=gray!35,
    boxrule=0.4pt,
    arc=2mm
]
We again use graph distances and slightly overload our notation from previous sections.
Here, $\texttt{dist}_R(i,j)$ is the distance between the index leaves $i,j$ in $R$ which is the complete binary tree, and
$\texttt{dist}_T(i,j)$ is the distance between the value leaves $X[i],X[j]$ in $T$.
Recall that a finger function is a map $f:[n]\setminus\{1\}\to[n]$ with $f(i)<i$.
\end{tcolorbox}
\begin{lemma}\label{lemm:quicksort_analysis}
Fix a random seed and let $R$ and $T$ be as above.
For any finger function $f$, the running time of \textsc{SkeletonSort} is
$O\bigl(n+\sum_{i=2}^{n}(\texttt{dist}_R(i,f(i))+\texttt{dist}_T(i,f(i)))\bigr)$.
\end{lemma}
\begin{proof}
Consider a non-empty invocation of \textsc{SkeletonSort} on a skeleton tree $S$,
where \textsc{Random-Leaf} selects leaf $i$ and pivot $p=X[i]$. The subsequent \textsc{Skeleton-Split} must traverse the entire path
from the root of $S$ to $i$, since every subtree along this path
contains $p$.
Thus, the time spent on \textsc{Random-Leaf}  is bounded by the time spent on \textsc{Skeleton-Split}.

What remains is to bound that time. 
A node \emph{gets split} if none of the three early stopping conditions apply (in those cases, the call takes constant time and does not recurse).
If $s$ nodes get split, the partitioning procedure visits $2s+1$ nodes and takes $O(1+s)$ time.

Singleton invocations take $O(1)$ time each, so we consider only invocations with at least two items. Recall that, for our fixed seed, the recursion tree $T_0$ represents the execution of Quicksort that \textsc{SkeletonSort} simulates.
Let $r$ be the node of $T_0$ corresponding to the current invocation $\textsc{SkeletonSort}(S)$, and let $h(r)$ denote its height in $T_0$.
Color a node of $S$ red if its raw height is at least $h(r)$, and blue otherwise.
Since $h(r)\geq1$, all leaves are blue.
If present, the red nodes form a connected subtree containing the root of $S$, with maximal blue subtrees hanging below it. If there are no red nodes, the entire tree $S$ is the sole maximal blue subtree.
If there are $b$ maximal blue subtrees, there are $b-1$ red nodes, as every red node has two children.

The following analysis is very similar to running \textsc{SproutSort}, and specifically, \cref{theo:grow_merge}, in reverse.
It might be useful to think of the red nodes as the root list of a sprout forest,
and each maximal blue subtree as a sprout tree.
(This analogy only works for the analysis. We cannot run \textsc{SkeletonSort} with sprout trees as, due to the random pivot selection, we do not know the height $h(r)$ of the current invocation during the algorithm.)

We charge the work of \textsc{Skeleton-Split} to leaf indices of $S$.
Define the \emph{leader} of a subtree of $S$ to be its leftmost leaf.
For each of the $b$ maximal blue subtrees, charge its leader once.
These $b$ charges cover the constant overhead and the cost of splitting all red nodes.
Next, consider a blue node $u$ that gets split by \textsc{Skeleton-Split}.
At least two of its three output trees are nonempty.
The leader of $u$ remains the leader of one of these output trees.
We charge the work at $u$ to the leader of another nonempty output tree.
Thus, every blue split is charged to an index that becomes a new leader.

We first consider charges to index $1$.
It receives at most one red charge per non-empty invocation of \textsc{SkeletonSort}, and therefore $O(n)$ red charges in total.
It receives no blue charges, since whenever its leaf is present, it is already the leader of its subtree.

To bound the remaining charges, fix $i\geq2$ and consider red charges to $i$.
If the current height $h(r)$ exceeds both
$\texttt{dist}_T(i,f(i))$ and $\texttt{dist}_R(i,f(i))/2$,
then the current input $S$ contains both indices $i$ and $f(i)$, and their common ancestor in $R$ is blue.
They therefore lie in the same maximal blue subtree.
Since $f(i)<i$, index $i$ cannot be its leader and receives no red charge.
The heights decrease strictly along the recursive path containing $i$,
and $i$ receives at most one red charge per invocation.
Hence $i$ receives $O(1+\texttt{dist}_R(i,f(i))+\texttt{dist}_T(i,f(i)))$ red charges.

Next consider a blue split charged to $i$, and let $u$ be the corresponding
original node in $R$, of raw height $k$.
If the leaf $f(i)$ is not a descendant of $u$ in $R$, then
the lowest common ancestor of $i$ and $f(i)$ is a strict ancestor of $u$,
so $k < \texttt{dist}_R(i,f(i))/2$.
Otherwise, consider the output subtree whose new leader is $i$.
Since $f(i)<i$, this subtree cannot contain the leaf $f(i)$.
Thus, the current or an earlier Quicksort pivot must have separated
$X[i]$ and $X[f(i)]$.
Their lowest common ancestor in $T_0$ is therefore the current
invocation or one of its ancestors, so
$h(r) \leq \texttt{dist}_T(i,f(i))$.
Since the charged node is blue, $k<h(r)$.
In either case, $k<\max\{\texttt{dist}_R(i,f(i)) / 2,\texttt{dist}_T(i,f(i))\}$.

Once a blue split at an original node $u$ is charged to $i$,
the leaf $i$ becomes the leftmost surviving leaf below $u$ in its output tree.
Every subsequent invocation containing $i$ operates on a subset of these leaves,
so $i$ remains the leftmost surviving leaf below $u$.
In other words, $i$ is only charged when it becomes the leader of a node $u$,
after which $i$ remains the leader of that node as long as $u$ remains blue.
Thus, $i$ receives at most one blue charge at each original node $u \in R$.
Since there is at most one ancestor of $i$ at each raw height, and all charged nodes have raw height
$k<\max\{\texttt{dist}_R(i,f(i))/2,\texttt{dist}_T(i,f(i))\}$,
index $i$ receives
$O(1+\texttt{dist}_R(i,f(i))+\texttt{dist}_T(i,f(i)))$
blue charges in total.

Summing the red and blue charges over all indices,
and adding the $O(n)$ initialization and output costs, proves the lemma.
\end{proof}

\begin{lemma}\label{lemm:random_LCA}
Let $X$ be a list of $n$ distinct elements and let $T$ be
as above.
For any distinct indices $i,j\in[n]$, we have
$\mathbb{E}[\texttt{dist}_T(i,j)]
\in O(1+\log|\rank_X(X[i])-\rank_X(X[j])|)$.
\end{lemma}
\begin{proof}
We use the model of random binary search trees~\cite{SeidelAragon1996} where each rank
$x\in[n]$ receives an independent uniform priority $p_x\in[0,1]$,
and the root of every subtree is the rank with the largest priority.
This has the same distribution as the Quicksort recursion tree $T_0$.

Let $a<b$ be the ranks of $X[i]$ and $X[j]$ in sorted order,
and put $d=b-a$.
Note that $d=|\rank_X(X[i])-\rank_X(X[j])|$.
Let $k$ be the rank with the largest priority among $[a,b]$.
Then the invocation with pivot $k$ is the lowest common ancestor
of $X[i]$ and $X[j]$ in $T_0$.

Let $B$ be the first rank greater than $b$ whose priority
exceeds $p_k$ (or $n+1$ if no such rank exists).
Since $k$ has the largest priority in $[a,b]$,
the right subtree of $k$ contains exactly the elements with rank
$k+1,\dots,B-1$.
Conditioned on $k$ and $B$, this is a random binary search tree.
The expected height of a random binary
search tree with $s$ keys is $\le 5 \log(2+s)$, see \cite[Chapter 12.4]{CLRS2009}.
Thus, the expected height of the right subtree, conditioned on $k$ and $B$, is
$\le 5\log(2+B-a)$.
For any $t\in[b,n]$, the event $B>t$ occurs precisely when
the largest priority among $[a,t]$ belongs to $[a,b]$.
Thus,
$\Pr(B>t)=(d+1)/(t-a+1)$.
In particular, for every integer $q\geq0$,
$\Pr(B-a>2^q(d+1))\leq 2^{-q}$.
Summing over these exponentially increasing ranges gives
$\mathbb{E}[\log(2+B-a)]\in O(1+\log d)$.
Hence the expected height of the right subtree is $O(1+\log d)$.
By a symmetric argument, the expected height of the left subtree is also $O(1+\log d)$.
The height of $k$ is one larger than the maximum height of the two subtrees,
thus $\mathbb{E}[h(k)] \in O(1+\log d)$.

Recall that $T$ is obtained from $T_0$ by extending edges
until all leaves have the same depth.
Thus, $\texttt{dist}_T(i,j) = 2 h(k)$.
Therefore, $\mathbb{E}[\texttt{dist}_T(i,j)]\in O(1+\log d)$,
which proves the lemma.
\end{proof}

\begin{lemma}\label{lemm:skeleton_distance}
Let $X$ be a list of $n$ distinct elements and let $f$ be any finger function.
Then
\begin{align*}
\sum_{i=2}^{n}\texttt{dist}_R(i,f(i))
&\in O\left(n+\sum_{i=2}^{n}\log(i-f(i))\right),\\
\mathbb{E}\left[\sum_{i=2}^{n}\texttt{dist}_T(i,f(i))\right]
&\in O\left(n+\sum_{i=2}^{n}
\log \max\{1, |\rank_X(X[i])-\rank_X(X[f(i)])|\}\right).
\end{align*}
\end{lemma}
\begin{proof}
The first bound follows directly from \cref{lemm:complete_distance},
since $R$ is the complete binary tree whose leaves are the
indices $1,\dots,n$ in-order. For the second bound, we apply linearity of expectation and \cref{lemm:random_LCA}.
\end{proof}

\begin{theorem}\label{theo:skeleton_sort}
For any list $X$ of $n$ distinct elements, \textsc{SkeletonSort}
sorts $X$ in expected $O(n+\log\reg(X))$ time.
\end{theorem}
\begin{proof}
Let $f$ be an arbitrary finger function.
By \cref{lemm:quicksort_analysis}, for every fixed random seed,
the running time of \textsc{SkeletonSort} is
$O(n+\sum_{i=2}^{n}
(\texttt{dist}_R(i,f(i))+\texttt{dist}_T(i,f(i))))$.

Taking expectations and applying \cref{lemm:skeleton_distance},
the expected running time is
\begin{align*}
O\left(
n+\sum_{i=2}^{n}\log(i-f(i))
+\sum_{i=2}^{n}\log\max\{1, |\rank_X(X[i])-\rank_X(X[f(i)])|\}
\right).
\end{align*}
By the definition of $\reg$, there exists a finger function $f$
for which these two sums equal $\log\reg(X)$.
The theorem follows.
\end{proof}

\newpage
\section{Notions of presortedness}

{
\footnotesize
\setlength{\tabcolsep}{3.5pt}
\renewcommand{\arraystretch}{1.22}
\setlength{\LTleft}{-0.15\linewidth}
\setlength{\LTright}{-0.075\linewidth}
\begin{longtable}{@{}>{\raggedright\arraybackslash}p{0.207\linewidth}
                    >{\raggedright\arraybackslash}p{0.276\linewidth}
                    >{\raggedright\arraybackslash}p{\dimexpr0.8671\linewidth-5.2\tabcolsep\relax}@{}}
\caption{Measures of presortedness and their target adaptive bounds. Throughout, $X:[n]\to[n]$ is surjective onto $[\max(X)]$, $\ol{X}$ is its stable rank permutation, and logarithms are base two. Zero-frequency entropy terms contribute zero. The notation $X=X_1\sqcup\cdots\sqcup X_k$ denotes a partition into disjoint, possibly interleaved subsequences.
For brevity, increasing means non-decreasing, and decreasing means non-increasing.}
\label{tab:presortedness}\\
\toprule
Name (aliases) & Adaptive bound & Mathematical definition + intuition \\
\midrule
\endfirsthead
\toprule
Name (aliases) & Adaptive bound & Mathematical definition + intuition \\
\midrule
\endhead
\midrule
\multicolumn{3}{r}{\emph{Continued on next page}}\\
\endfoot
\bottomrule
\endlastfoot

$\Ham$
  ($\mathrm{Pos}$) \cite{Mannila1985}
& $O(n+\Ham\log(\Ham+1))$
& $\displaystyle \Ham(X)=|\{i\in[n]\mid \ol{X}[i]\neq i\}|$.
  \par\smallskip
  \emph{Number of elements outside their final positions.} \\

$\Exc$
  \cite{Mannila1985}
& $O(n+\Exc\log(\Exc+1))$
& $\displaystyle \Exc(X)=n-c(X)$, where $c(X)$ is the number of cycles in $i\mapsto \ol{X}[i]$.
  \par\smallskip
  \emph{Minimum number of swaps to reach the stable sorted order; for distinct keys, the minimum number of swaps to sort $X$.} \\

$\Rem$
  ($\mathrm{Las}$) \cite{CookKim1980,Mannila1985}
& $O(n+\Rem\log(\Rem+1))$
& $\displaystyle \Rem(X)=\min\{|Y|:Y\subseteq X\text{ and }X-Y\text{ is increasing}\}=n-\mathrm{Las}(X)$.
  \par\smallskip
  \emph{Minimum number of deletions needed to leave $X$ sorted.} \\

\midrule

$\Max$
  \cite{Mannila1985}
& $O\!\left(n\bigl(1+\log(\Max+1)\bigr)\right)$
& $\displaystyle \Max(X)=\max_{i\in[n]}|\ol{X}[i]-i|$.
  \par\smallskip
  \emph{Maximum distance from an element's final position.} \\

$\Par$
  ($\mathrm{Dis}$, $\mathrm{Radius}$) \cite{EstivillCastroWood1989,AltmanIgarashi1989}
& $O\!\left(n\bigl(1+\log(\Par+1)\bigr)\right)$
& $\displaystyle \Par(X)=\max\{j-i:1\leq i<j\leq n,\ X[i]>X[j]\}$, with $\max\emptyset=0$.
  \par\smallskip
  \emph{Maximum positional distance between an inverted pair.} \\

$\Inv$ \cite{Mannila1985}
& $O\!\left(n+n\log\!\left(1+\frac{\Inv}{n}\right)\right)$
& $\displaystyle \Inv(X)=|\{(i,j):1\leq i<j\leq n,\ X[i]>X[j]\}|$.
  \par\smallskip
  \emph{Number of inverted pairs.} \\

$\DS$
  ($\mathrm{Spear}$) \cite{DiaconisGraham1977}
& $O\!\left(n+n\log\!\left(1+\frac{\DS}{n}\right)\right)$
& $\displaystyle \DS(X)=\sum_{i=1}^n |\ol{X}[i]-i|$.
  \par\smallskip
  \emph{Total distance from elements' final positions.} \\

$\HDS$
  ($\mathrm{HD}$) \cite{BaiCoester2023}
& $O(n\HDS)$
& $\displaystyle \HDS(X)=\frac1n\sum_{i=1}^n
    \log\bigl(|\ol{X}[i]-i|+2\bigr)$.
  \par\smallskip
  \emph{Average logarithmic distance from elements' final positions.} \\

\midrule
$\nMulti$
   \cite{MunroSpira1976}
& $O(n \log (1+\nMulti))$
& $\nMulti = \max(X) = |\{x \in X\}|$.
  \par\smallskip
  \emph{Number of distinct values.} \\
  
$\HMulti$
   \cite{MunroSpira1976}
& $O(n (1+\HMulti))$
& $\HMulti = \sum_{i=1}^{n} \frac{m_i}{n} \log \frac{n}{m_i}$, where $m_i$ is the number of occurrences of $i$ in $X$.
  \par\smallskip
  \emph{Entropy of the value multiplicities.} \\
  
\midrule

$\Block$
  \cite{CarlssonLevcopoulosPetersson1990}
& $O(n+\Block\log \Block)$
& $\displaystyle \Block(X)=1+|\{i\in[n-1]:\ol{X}[i+1]\neq \ol{X}[i]+1\}|$.
  \par\smallskip
  \emph{Number of maximal sorted blocks contiguous in both position and value. The measure $\nSRuns$ uses the same count but the bound $O(n(1+\log\nSRuns))$.} \\

$\HSRuns$ ($\HBlock$)
  ($H(\vSRuns)$) \cite{BarbayNavarro2013}
& $O\!\left(n(1+\HSRuns)\right)$
& $\displaystyle \HSRuns(X)=\HBlock(X)=\sum_{i=1}^{k}\frac{b_i}{n}\log\frac{n}{b_i}$,
  where $b_1,\dots,b_k$ are the lengths of the maximal increasing blocks
  that are contiguous in both index and value.
  \par\smallskip
  \emph{Entropy of lengths of maximal sorted blocks contiguous in both position and value.} \\

$\nRuns$ \cite{Knuth1998,Mannila1985}
& $O\!\left(n(1+\log\nRuns)\right)$
& $\displaystyle \nRuns(X)=1+|\{i\in[n-1]:X[i+1]<X[i]\}|$.
  \par\smallskip
  \emph{Number of maximal sorted contiguous pieces.} \\

$\HRuns$
  ($H(\vRuns)$, runlength) \cite{McIlroy1993,BarbayNavarro2013}
& $O\!\left(n(1+\HRuns)\right)$
& $\displaystyle \HRuns(X)=\sum_{i=1}^{k}\frac{b_i}{n}\log\frac{n}{b_i}$,
  where $b_1,\dots,b_k$ are the lengths of the maximal increasing
  index-contiguous blocks.
  \par\smallskip
  \emph{Entropy of lengths of maximal sorted contiguous pieces.} \\

$\nRuns^*$
  (dual runs, riffles) \cite{McIlroy1993}
& $O\!\left(n(1+\log\nRuns^*)\right)$
& $\displaystyle \nRuns^*(X)=1+|\{v\in[n-1]:\ol{X}^{-1}[v+1]<\ol{X}^{-1}[v]\}|$. For permutations, this equals $\nRuns(X^{-1})$.
  \par\smallskip
  \emph{Number of maximal consecutive value ranges appearing in sorted order.} \\

$\HRuns^*$
  (dual runs, riffles) \cite{McIlroy1993,GhasemiJugeKhalighinejad2022}
& $O\!\left(n(1+\HRuns^*)\right)$
& $\displaystyle \HRuns^*(X)=\sum_{i=1}^{k}\frac{b_i}{n}\log\frac{n}{b_i}$,
  where $b_1,\dots,b_k$ are the lengths of the maximal increasing
  value-contiguous subsequences.
  \par\smallskip
  \emph{Entropy of lengths of maximal consecutive value ranges appearing in sorted order.} \\

\midrule

$\nSUS$
  ($\mathrm{SUS}$, $\mathrm{Lds}$) \cite{LevcopoulosPetersson1994}
& $O\!\left(n(1+\log(\nSUS))\right)$
& $\displaystyle \nSUS(X)=\min\{k:X=X_1\sqcup\cdots\sqcup X_k,\ X_r\text{ increasing}\}$.
  \par\smallskip
  \emph{Minimum number of increasing subsequences partitioning $X$.} \\

$H(\vSUS)$
  (greedy SUS) \cite{BarbayNavarro2013}
& $O\!\left(n\bigl(1+H(\vSUS)\bigr)\right)$
& $\displaystyle H(\vSUS(X))=\sum_{i=1}^{k}\frac{b_i}{n}\log\frac{n}{b_i}$,
  where $b_1,\dots,b_k$ are the lengths of the increasing subsequences
  obtained by greedily scanning $X$ left to right and appending each element to
  the subsequence with the largest last element at most the new element,
  starting a new subsequence if none exists.
  \par\smallskip
  \emph{Entropy of subsequence lengths in the greedy increasing partition.} \\

$H(\vLRM)$ \cite{BarbayFischer2012}
& $O\!\left(n\bigl(1+H(\vLRM)\bigr)\right)$
& $\displaystyle H(\vLRM(X))=\sum_{i=1}^{k}\frac{b_i}{n}\log\frac{n}{b_i}$,
  where $b_1,\dots,b_k$ are the lengths of the subsequences obtained
  by repeatedly removing a longest root-to-leaf path from the remaining
  LRM forest, excluding the dummy root.
  \par\smallskip
  \emph{Entropy of subsequence lengths obtained by joining runs along the LRM tree.} \\

$\HSUS$
  (chromatic entropy) \cite{BarbayNavarro2013,AlonOrlitsky1996}
& $O\!\left(n(1+\HSUS)\right)$
& $\displaystyle \HSUS(X)=\min_{X=X_1\sqcup\cdots\sqcup X_k}
  \sum_{i=1}^{k}\frac{|X_i|}{n}\log\frac{n}{|X_i|}$,
  where each $X_i$ is increasing.
  \par\smallskip
  \emph{Minimum entropy of subsequence lengths over increasing partitions.} \\

$\nSMS$
  ($\mathrm{SMS}$, $\mathrm{Lads}$) \cite{LevcopoulosPetersson1994,Wagner1984}
& $O\!\left(n(1+\log(\nSMS))\right)$
& $\displaystyle \nSMS(X)=\min\{k:X=X_1\sqcup\cdots\sqcup X_k,\ X_r\text{ increasing or decreasing}\}$.
  \par\smallskip
  \emph{Minimum number of monotone subsequences partitioning $X$.} \\

$H(\vSMS)$ \cite{BarbayNavarro2013}
& $O\!\left(n\bigl(1+H(\vSMS)\bigr)\right)$, given the partition
& $\displaystyle H(\vSMS(X))=\sum_{i=1}^{k}\frac{b_i}{n}\log\frac{n}{b_i}$,
  where $b_i=|X_i|$ for a chosen partition $X=X_1\sqcup\cdots\sqcup X_k$
  into the minimum number of monotone subsequences.
  \par\smallskip
  \emph{Entropy of subsequence lengths in a specific monotone partition.} \\

$\HSMS$ \cite{BarbayNavarro2013}
& $O\!\left(n(1+\HSMS)\right)$
& $\displaystyle \HSMS(X)=\hspace{-0.3cm}\min_{X=X_1\sqcup\cdots\sqcup X_k}
  \sum_{i=1}^{k}\frac{|X_i|}{n}\log\frac{n}{|X_i|}$,
  with each $X_i$ increasing or decreasing.
  \par\smallskip
  \emph{Minimum entropy of subsequence lengths over monotone partitions.} \\

$\mathrm{enc}$ \cite{Skiena1988}
& $O\!\left(n(1+\log(\mathrm{enc}))\right)$
& $\displaystyle \mathrm{enc}(X)=k$,
  where $k$ is the number of sorted lists built by scanning $X$ left to right
  and adding each element at the beginning or end of the earliest-created list
  for which this preserves increasing order, creating a new list if none exists.
  \par\smallskip
  \emph{Number of sorted lists built greedily by insertion at either end.} \\

\midrule

$\osc$~\cite{LevcopoulosPetersson1993}
& $O\!\left(n\left(1+\log\left(1+\frac{\osc}{n}\right)\right)\right)$
&$\displaystyle
\osc(X)=
\sum_{i=1}^{n-1}
\left(
\max\{0,\left|
X[i+1]
-X[i]
\right|-1\}
\right)$.
\par\smallskip
\emph{Total number of distinct ranks strictly between consecutive values; the formula extends the permutation measure to repeated keys.}
\\

$\Loc$ \cite{KatajainenLevcopoulosPetersson1989} (dynamic finger bound \cite{ProofreadCole2000})
& $O(n+\log(\Loc))$
& $\displaystyle \Loc(X)=\prod_{i=2}^n \max \{ 1, 
  \left|\operatorname{rank}_{X[1..i]}(X[i])
  -\operatorname{rank}_{X[1..i]}(X[i-1])\right|\}$.
  \par\smallskip
  \emph{Consecutive elements should be close in rank among elements seen so far.} \\

$\Hist$ \cite{MoffatPetersson1991} (working-set property \cite{SleatorTarjan1985})
& $O(n+\log(\Hist))$
& $\displaystyle \Hist(X)=\prod_{i=2}^n
  \min\left\{i-j:j<i \land
  \left|\operatorname{rank}_{X[1..i]}(X[i])
  -\operatorname{rank}_{X[1..i]}(X[j])\right|\le 1\right\}$.
  \par\smallskip
  \emph{Elements should have a recently inserted predecessor or successor.} \\

$\Reg$ \cite{MoffatPetersson1991,PeterssonMoffat1995} (unified bound~\cite{BadoiuColeDemaineIacono2007,DerryberrySleator2009})
& $O(n+\log(\Reg))$
& $\displaystyle
\Reg(X)=\prod_{i=2}^n
\min_{j<i}\left\{
(i-j)\cdot
\max\left\{
1,
\left|
\operatorname{rank}_{X[1..i]}(X[i])
-\operatorname{rank}_{X[1..i]}(X[j])
\right|
\right\}
\right\}.
$
  \par\smallskip
  \emph{Elements should be close to an earlier element in position and current rank.} \\

\midrule

$\loc$
  \cite{McIlroy1993}
& $O(n+\log(\loc))$
& $\displaystyle \loc(X)=\prod_{i=2}^{n} 
\max \left\{1,
  \left|X[i]
  -X[i-1]\right| \right\}$.
  \par\smallskip
  \emph{Consecutive elements should be close in final rank.} \\

$\hist$
  \cite{McIlroy1993}
& $O(n+\log(\hist))$
& $\displaystyle \hist(X)= \loc(\ol{X}^{-1}) = \prod_{i=2}^{n} |\ol{X}^{-1}[i] - \ol{X}^{-1}[i-1]|$,
  \par\smallskip
  \emph{Elements of consecutive final ranks should be close in position.} \\

$\reg$
  (this paper)
& $O(n+\log(\reg))$
& $\displaystyle \reg(X)=\prod_{i=2}^{n}\min_{j<i}\left\{(i-j)\max\{1,|X[i]-X[j]|\}\right\}$.
  \par\smallskip
  \emph{Each new element should be close to an earlier element in position and final rank.} \\
  
$\unimst$
  (this paper)
& $O(n+\unimst)$
& $\displaystyle \unimst(X)=\text{weight of an MST with }w(i,j)=\log(2+|i-j|+|X[i]-X[j]|)$.
  \par\smallskip
  \emph{Minimum spanning-tree weight for logarithmic distance in the index--rank plane.} \\
  
  
  

\end{longtable}
}

\section{The Hierarchy of Adaptive Measures} \label{sect:arrows}

\begin{figure}[h]
    \centering
\includegraphics[]{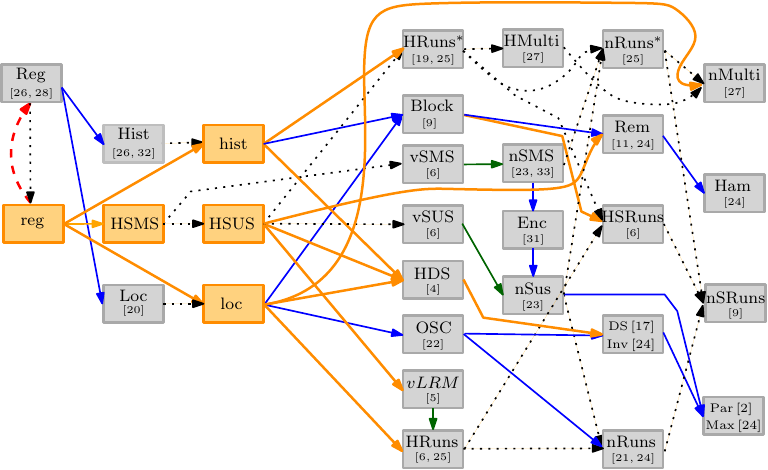}
\caption{We recall Figure~\ref{fig:hierarchy}. Orange and red items are new. Red only holds if $X$ has no duplicates. Blue arrows are attributed to~\cite{MoffatPetersson1991} and green edges to~\cite{BarbayNavarro2013}.} \label{fig:hierarchy_proofs}
\end{figure}

In \cref{sect:reg_relations}, we show that $\reg \preceq \hist$, $\reg \preceq \loc$ and $\reg \preceq \HSMS$.
In this section, we show the remaining relations in \cref{fig:hierarchy}, as depicted in \cref{fig:hierarchy_proofs}.
Some of these relations were already shown by Moffat and Petersson~\cite{MoffatPetersson1991}, or Blelloch and Dobson \cite{BlellochDobson2023}.
While Moffat and Petersson use $\Hist$ and $\Loc$, their proofs also work for $\hist$ and $\loc$.

Many relations in \cref{fig:hierarchy_proofs} are straightforward from the definition.
For example, $\nSUS$ considers partitions that are less restrictive than those of $\nRuns$, which in turn is less restrictive than $\nSRuns$; and all of these compute the same adaptive bound from a given permutation.
Therefore $\nSUS \preceq \nRuns \preceq \nSRuns$.
Another example: $\HSUS$ minimises the entropy over a family of partitions that includes the one partition considered by $H(\vSUS)$.
Therefore $\HSUS \preceq H(\vSUS)$.
Furthermore, as explained in the introduction, entropy-based measures are stronger than their counting variant, e.g. $\HRuns \preceq \nRuns$. (See also \cref{lemm:entropy_bounds} (1) below.)
We explicitly prove the remaining relations in \cref{fig:hierarchy_proofs} in this section.

\subparagraph{Inaccuracies in Barbay-Navarro} Figure~1 in \cite{BarbayNavarro2013} claims the following relations, which seem to not hold: First they claim that $\Block \cong \nSRuns$.
However, these notions use different target running times: $O(n + \Block \log(1+\Block))$ vs $O(n \log (1+\nSRuns))$.
Second, they claim that $H(\vSUS) \preceq H(\vLRM)$, but Section 4 and Lemma 7 in~\cite{BarbayFischer2012} show that these are incomparable.
Third, they claim that $H(\vSMS) \preceq H(\vSUS)$.
This holds if $H(\vSMS)$ denotes the minimum entropy over all SMS partitions, but need not hold if $\vSMS$ is an unspecified minimum-cardinality partition, which may have much larger entropy than their greedy $\vSUS$ partition.
What is true, however, is $\HSMS \preceq \HSUS$.

\subsection{Proving the remaining relations}

We first prove a helper lemma on entropy functions.
\begin{lemma} \label{lemm:entropy_bounds}
    Let $b_1, \dots, b_k$ be positive integers with $b_1 + \dots + b_k = n$.
    Let $H(b_1, \dots, b_k) = \sum_{i=1}^{k} \frac{b_i}{n} \log \frac{n}{b_i}$ be the entropy function. Then
    \begin{align}
        n H(b_1, \dots, b_k) &\le n \log k\\
        k \log n \le 2 n + k \log k &\in O(n (1 + H(b_1, \dots, b_k)))\\
        n H(b_1, \dots, b_k) &\le n + 2 \sum_{i=1}^{k} b_i \log i
    \end{align}
\end{lemma}
\begin{proof}
    By concavity of $\log$, we have
    \[
    H(b_1, \dots, b_k) = \sum_{i=1}^{k} \frac{b_i}{n} \log \frac{n}{b_i} \le \log \Big(\sum_{i=1}^{k} \frac{b_i}{n} \cdot \frac{n}{b_i} \Big) = \log k.
    \]
    This shows (1). Next, we show (2). By re-indexing, we may assume that $b_1 \ge b_2 \ge \dots \ge b_k$.
    Then $\frac{n}{b_i} \ge i$ since $n \ge b_1 + \dots + b_i \ge i \cdot b_i$.
    Also, $b_i \ge 1$ since $b_i$ is positive. Thus,
    \[
    n H(b_1, \dots, b_k) = \sum_{i=1}^{k} b_i \log \frac{n}{b_i} \ge \sum_{i=1}^{k} 1 \cdot \log i = \log \prod_{i=1}^{k} i = \log k! \in \Omega(k \log k).
    \]
    In particular, $n + k \log k \in O(n(1+H(b_1, \dots, b_k)))$.
    Finally, $\log \frac{n}{k} \le \log e \cdot \frac{n}{k}$ and thus
    \[
    k \log n = (k \log k + k \log \frac{n}{k}) \le k \log k + k \log e \cdot \frac{n}{k} \le k \log k + 2 n.
    \]
    This shows (2). Next, we show (3).
    The log-sum inequality for sequences $y_1, \dots, y_k > 0$ and $z_1, \dots, z_k > 0$ with $Y = y_1 + \dots + y_k$ and $Z = z_1 + \dots + z_k$ states that
    \[
    \sum_{i=1}^{k} y_i \log \frac{1}{y_i} \le \sum_{i=1}^{k} y_i \log \frac{1}{z_i} + Y \log \frac{Z}{Y}.
    \]
    We apply this with $y_i :=\frac{b_i}{n}$ and $z_i := \frac{1}{i^2}$, which yields $Y = 1$, $Z \le \frac{\pi^2}{6}$, and
    \[
    H(b_1,\dots,b_k)=\sum_{i=1}^{k}\frac{b_i}{n}\log\frac{n}{b_i}\leq\sum_{i=1}^{k}\frac{b_i}{n}\log(i^2)+\log\frac{\pi^2}{6}\leq2\sum_{i=1}^{k}\frac{b_i}{n}\log i+1.
    \]
    This shows (3).
\end{proof}
\cref{lemm:entropy_bounds} (1) immediately implies the ``entropy'' arrows in \cref{fig:hierarchy_proofs}.

\begin{lemma}
    $\Block \preceq \HSRuns$.
\end{lemma}
\begin{proof}
    Let $X = X_1 \sqcup \dots \sqcup X_k$ be the partition of $X$ into maximal increasing sequences that are both contiguous in index and value. By \cref{lemm:entropy_bounds} (2),
    \[
    n + \Block \log (1+\Block) = n + k \log (k+1) \in O(n(1 + H(|X_1|, \dots, |X_k|))) = O(n(1+\HSRuns)).
    \]
    This shows $\Block \preceq \HSRuns$.
\end{proof}
\begin{lemma}
    $\HDS \preceq \DS$.
\end{lemma}
\begin{proof}
By concavity of $x \mapsto \log(2+x)$, we have
\[
\HDS = \sum_{i=1}^{n} \frac{\log (2+|\ol{X}[i] - i|)}{n} \le \log (2 + \frac{\sum_{i=1}^{n} |\ol{X}[i] - i|}{n}) = \log (2+\frac{\DS}{n}).
\]
In particular, $n \HDS \in O(n + n \log(1+\frac{\DS}{n}))$. This shows $\HDS \preceq \DS$.
\end{proof}
\begin{lemma} \label{lemm:loc_runs}
    $\loc \preceq \HRuns$ and $\hist \preceq \HRuns^*$.
\end{lemma}
\begin{proof}
    Let $X = X_1 \sqcup \dots \sqcup X_k$ be the partition of $X$ into maximal increasing sequences that are contiguous in index. Put $s_i = |X_i|$.
    A singleton run contributes zero. For $s_i>1$, monotonicity gives $\sum_{j=2}^{s_i}\max\{1,X_i[j]-X_i[j-1]\}\leq s_i-1+n$. By concavity of $\log$,
    \[
    \sum_{j=2}^{s_i}\log\max\{1,|X_i[j]-X_i[j-1]|\}\leq(s_i-1)\log\left(1+\frac{n}{s_i-1}\right)\leq s_i\log\frac{2n}{s_i}.
    \]
    Due to the boundaries between the $X_i$, the sum $\log \loc$ contains $k-1$ additional terms that are each at most $\log n$.
    Hence, by \cref{lemm:entropy_bounds} (2), we have
    \[
    \log\loc\leq k\log n+\sum_{i=1}^{k}s_i\log\frac{2n}{s_i}=k\log n+n+\sum_{i=1}^{k}s_i\log\frac{n}{s_i}\in O(n(1+H(s_1,\dots,s_k))).
    \]
    This shows that $\loc \preceq \HRuns$.
    For $\hist$, observe that $\hist(X) = \loc(\ol{X}^{-1})$.
    On the other hand, $\HRuns^*(\ol{X}) \le \HRuns^*(X)$.
    Indeed, any value-contiguous subsequence of $X$ contains, for each value, either all copies or no copies of that value.
    Thus, a value-contiguous subsequence of $X$ is also a value-contiguous subsequence of $\ol{X}$.
    The partition of $\ol{X}$ into maximal such subsequences therefore only merges parts, and entropy does not increase under merging.
    Applying $\loc \preceq \HRuns$ to $\ol{X}^{-1}$ now shows that $\hist \preceq \HRuns^*$.
\end{proof}
\begin{lemma}
    $\loc \preceq \HDS$ and $\hist \preceq \HDS$.
\end{lemma}
\begin{proof}
    For positive integers $a,b,c$, by the triangle inequality, we have
    \[
    2 + |c-a| \le 3 + |a-b| + |c-(b+1)| \le (2+|a-b|) \cdot (2+|c-(b+1)|).
    \]
    We use this with $a = \ol{X}[i-1], b=i-1, c = \ol{X}[i]$, and use the fact that $|X[i] - X[i-1]| \le |\ol{X}[i] - \ol{X}[i-1]|$ to get
    \[
    \loc \le \prod_{i=2}^{n} (2+|\ol{X}[i] - \ol{X}[i-1]|) \le \prod_{i=2}^{n} \Big((2+|\ol{X}[i-1] - (i-1)|) \cdot (2+|\ol{X}[i] - i|)\Big) \le \prod_{i=1}^{n} (2+|\ol{X}[i]-i|)^2.
    \]
    Taking logarithms yields $\log \loc \le 2 n \HDS$, which shows $\loc \preceq \HDS$.
    Since $\hist(X) = \loc(\ol{X}^{-1})$ and $\HDS(X) = \HDS(\ol{X}) = \HDS(\ol{X}^{-1})$, this also shows $\hist \preceq \HDS$.
\end{proof}
\begin{lemma}
    $\HSUS \preceq \HDS$.
\end{lemma}
\begin{proof}
    Apply the following construction to the permutation $P:=\ol{X}$. Any increasing subsequence of $P$ corresponds to a non-decreasing subsequence of $X$, and $\HDS(P)=\HDS(X)$.
    For each integer $d$, let $Y_d$ be the subsequence consisting of entries $P[i]$ with $P[i]-i=d$. Each $Y_d$ is increasing. Order these subsequences as $Y_0,Y_1,Y_{-1},Y_2,Y_{-2},\dots$ and discard empty parts, obtaining $X_1,\dots,X_k$. An entry $P[i]\in X_j$ then satisfies $j\leq2(2+|P[i]-i|)$.
    By \cref{lemm:entropy_bounds} (3),
    \[
    nH(|X_1|,\dots,|X_k|)\leq n+2\sum_{j=1}^{k}|X_j|\log j\leq n+2\sum_{i=1}^{n}\log(2(2+|P[i]-i|))\in O(n(1+\HDS(X))).
    \]
    Transferring this partition to $X$ proves $\HSUS\preceq\HDS$.
\end{proof}
\begin{lemma}
    $\HSUS \preceq \Rem$.
\end{lemma}
\begin{proof}
    Let $k = \Rem+1$. We partition $X = X_1 \sqcup \dots \sqcup X_k$ where $X_2, \dots, X_k$ are singletons whose removal leaves the sorted sequence $X_1$.
    By \cref{lemm:entropy_bounds} (2),
    \[
    nH(|X_1|,\dots,|X_k|)=|X_1|\log\frac{n}{|X_1|}+\sum_{i=2}^{k}\log n\leq2n+k\log n\leq4n+k\log k.
    \]
    Thus, $n(1+\HSUS) \in O(n + \Rem \log (1+\Rem))$, which shows $\HSUS \preceq \Rem$.
\end{proof}
\begin{lemma} \label{lemm:Loc_loc}
    $\Loc \preceq \loc$ and $\Reg\preceq\reg$.
\end{lemma}
\begin{proof}
    We have $\left|\rank_{X[1..i]}(X[i])
  -\rank_{X[1..i]}(X[i-1])\right|\leq|X[i]-X[i-1]|$, and thus
    \begin{align*}
    \Loc(X)&=\prod_{i=2}^{n}\max\{1,|\rank_{X[1..i]}(X[i])-\rank_{X[1..i]}(X[i-1])|\}\\
    &\leq\prod_{i=2}^{n}\max\{1,|X[i]-X[i-1]|\}=\loc(X).
    \end{align*}
    This shows $\Loc \preceq \loc$.
  The same rank-distance argument immediately gives $\Reg\preceq\reg$.
\end{proof}
\begin{lemma}
    $\loc \preceq \nMulti$.
\end{lemma}
\begin{proof}
    We have
    \[
    \loc(X)=\prod_{i=2}^{n}\max\{1,|X[i]-X[i-1]|\}\leq\max(X)^{n-1}\leq\max(X)^n.
    \]
    Taking logarithms yields $n + \log \loc \in O(n \log(1+\nMulti))$, which shows $\loc \preceq \nMulti$.
\end{proof}

\begin{lemma}
    $\Hist \preceq \hist$.
\end{lemma}
\begin{proof}
    If we replace $X$ by $\ol{X}$, then $\Hist$ only increases, while $\hist$ is unaffected.
    Thus, it suffices to consider the cases where $X$ is a permutation.
    
    For $i\in[n]$, let $S_i:=\{j\in[n]:j<i\land|\rank_{X[1..i]}(X[i])-\rank_{X[1..i]}(X[j])|\leq1\}$.
    Our goal is to define a function $f : [n] \setminus \{1\} \to [n]$ with $f(i) \in S_i$ for all $i$
    and $\sum_{i=2}^{n}\log(i-f(i))\in O(n+\log\hist(X))$.
    Given such a function, picking $j = f(i)$ in the definition of $\Hist$ shows $\Hist \preceq \hist$.

    For $t \in [n]$, the \emph{$t$-front} consists of all $i \in [n]$ with $X[i] \le t$ such that there exists no $j < i$ with $X[i] \le X[j] \le t$.
    
    \begin{claim}
        Let $t \in [n]$ be arbitrary.
        Put $I_t = X^{-1}([1..t])$ and $r_t = \min(I_t)$.
        There exists a function $f : I_t \setminus \{r_t\} \to I_t$
        with $f(i) \in S_i$ for all $i$ with the sum property
        \[
        \sum_{i \in I_t \setminus \{r_t\}} \log (i - f(i)) \le \sum_{y=2}^{t} \log |X^{-1}[y] - X^{-1}[y-1]|,
        \]
        and the following additional property:
        If $i_1<\dots<i_k$ are the indices of the $t$-front, then $f(i_j)=i_{j-1}$ for all $j\geq2$.
    \end{claim}
    \begin{proof}[Proof of the Claim]
        We do induction over $t$.
        The base case $t = 1$ is straightforward as $f$ is the trivial function and both sums in the sum property are empty.
        Let $t \ge 2$ be arbitrary. Suppose the claim holds for $t-1$ and yields a function $f_{t-1}$.
        We construct $f_t$ by slightly modifying $f_{t-1}$, keeping all choices unchanged unless stated otherwise.
        Let $i_1, \dots, i_k$ be the $(t-1)$-front.
        Put $i = X^{-1}[t]$.
        If $i_k < i$, then the $t$-front is $i_1, \dots, i_k, i$. Put $f_t(i) = i_k$. The sum property is maintained since $i_k = X^{-1}[t-1]$ and $i = X^{-1}[t]$, meaning $\log(i - f_t(i)) = \log |X^{-1}[t] - X^{-1}[t-1]|$.
        
        Otherwise, suppose that $i_j < i < i_{j+1}$ for some index $j$.
        We put $f_t(i) = i_j$ and $f_t(i_{j+1}) = i$, and keep $f_t(i') = f_{t-1}(i')$ for all other arguments $i'$.
        The $t$-front is $i_1, \dots, i_j, i$, so $f_t$ satisfies the additional property.
        We have
        \[
        i - f_t(i) = i - i_{j} \le i_{j+1} - i_j = i_{j+1} - f_{t-1}(i_{j+1})
        \]
        and
        \[
        i_{j+1} - f_t(i_{j+1}) = i_{j+1} - i \le i_k - i = X^{-1}[t-1] - X^{-1}[t].
        \]
        Therefore, the sum property for $(t-1)$ implies the sum-property for $t$:
        \begin{align*}
        &\sum_{i' \in I_{t} \setminus \{r_{t}\}} \log (i' - f_{t}(i'))\\
        =\ &\log(i-f_t(i))+\log(i_{j+1}-f_t(i_{j+1}))-\log(i_{j+1}-f_{t-1}(i_{j+1}))+\sum_{i'\in I_{t-1}\setminus\{r_{t-1}\}}\log(i'-f_{t-1}(i'))\\
        \le\ &\log|X^{-1}[t-1]-X^{-1}[t]|+\sum_{y=2}^{t-1}\log|X^{-1}[y]-X^{-1}[y-1]|\\
        =\ &\sum_{y=2}^{t}\log|X^{-1}[y]-X^{-1}[y-1]|.
        \end{align*}
        Finally, if $i<i_1$, then $i_1=r_{t-1}$ and $i=r_t$.
        We put $f_t(i_1) = i$. The sum property follows since
        \[
        \log(i_1-i)\leq\log(i_k-i)=\log|X^{-1}[t-1]-X^{-1}[t]|.
        \]
        In all three cases, the choices belong to the corresponding sets $S_i$: the elements of the $(t-1)$-front are prefix maxima among the values $1,\dots,t-1$, and $t$ is larger than all these values.
    \end{proof}
    The claim with $t=n$ immediately yields the desired function $f$.
    This shows $\Hist \preceq \hist$.
\end{proof}

\subsection{Reg equals reg for permutations}
\label{sub:regthm}

\cref{lemm:Loc_loc} shows $\Reg \preceq \reg$.
In \cref{sect:reg_mst} we show that, if $X$ contains duplicates, then $\reg \not\preceq \Reg$.
However, if $X$ is a permutation, then we can show that $\reg \preceq \Reg$.

\regthm*

\begin{lemma} \label{lemm:reg_threshold}
    Let $X$ be a permutation. For $i \ge 2$, put
    \begin{align*}
    d(i, j) &= |X[i] - X[j]|\\
    D(i, j) &= |\rank_{X[1..i]}(X[i]) - \rank_{X[1..j]}(X[j])|\\
    r(i) &:= \min_{j < i} \max\{(i-j), d(i, j)\} \\
    R(i) &:= \min_{j < i} \max\{(i-j), D(i, j)\}.
    \end{align*}
    For every integer $t \ge 1$, we have
    \[
    |\{i : r(i) \ge t^3\}| \le 5|\{i : R(i) \ge t\}| + O\big(\frac{n}{t}\big).
    \]
\end{lemma}

\begin{proof}[Proof of \cref{theo:Reg_reg} from \cref{lemm:reg_threshold}]
    We apply \cref{lemm:reg_threshold} for $t = 2^j$ where $0 \le j \le \lfloor\log n\rfloor$. This yields
    \[
    \underbrace{\sum_{j=0}^{\lfloor\log n\rfloor} |\{i : r(i) \ge 2^{3j}\}|}_{=:S}
    \le 5 \underbrace{\sum_{j=0}^{\lfloor\log n\rfloor} |\{i : R(i) \ge 2^j\}|}_{=:T} + O(n).
    \]
    The contribution of an index $i$ to $S$ is at least $\log(r(i))/3 - 1$,
    and the contribution to $T$ is at most $\log(R(i)) + 1$.
    Therefore, this shows
    \[
    \sum_{i=2}^{n} \log r(i) \le 3 S + O(n) \le 15 T + O(n) \le 30 \sum_{i=2}^{n}\log R(i) + O(n).
    \]
    We have $\log \reg \le 2\sum_{i=2}^{n} \log r(i)$.
    We also have $\log \Reg \ge \sum_{i=2}^{n}\log R(i)$.
    Indeed, for $j<i$, put $a=i-j$ and
    $
    b=|\rank_{X[1..i]}(X[i])-\rank_{X[1..i]}(X[j])|$.
    Since $X$ is a permutation, $b\ge1$.
    The rank of $X[j]$ can increase by at most $a$ between steps $j$ and $i$.
    If $X[i]>X[j]$, at most $a-1$ of these elements are smaller than $X[j]$, so $D(i,j)\le b+a-1$.
    If $X[i]<X[j]$, the rank of $X[j]$ increases by at least one, and again $D(i,j)\le b+a-1$.
    Thus, $
    D(i,j)\le a+b-1\le ab$.
    In particular, $R(i)\le ab$ for every $j<i$, proving the claimed inequality.
    Therefore, $\log \reg \le 60 \log \Reg + O(n)$.
\end{proof}

\begin{proof}[Proof of \cref{lemm:reg_threshold}]
    The claim is immediate for $t=1$ or $t^3>n$, so assume $2\le t\le n^{1/3}$.
    Put $F(i) := \arg \min_{j < i} \max\{(i-j), D(i, j)\}$.
    We partition the value space $[n]$ into \emph{buckets} of size $t^3$.
    The last bucket may be smaller.
    The item $X[i]$ goes into the bucket $\lfloor (X[i]-1) / t^3 \rfloor$.
    We say an index $i$ is \emph{early} if $X[i]$ is among the first $4 t^2$ items inserted in its bucket.
    Observe that the total number of early items is $O(\frac{n}{t^3} \cdot 4 t^2) = O(\frac{n}{t})$.

    Additionally, we partition the index space into \emph{epochs} of size $t^2$.
    The last epoch may be smaller.
    The item $X[i]$ goes into the epoch $\lfloor (i-1) / t^2 \rfloor$.
    We say a bucket is \emph{heavy} at the start of an epoch if it contains at least $2 t^2$ previously inserted items,
    and \emph{light} otherwise.

    Our proof is a charging argument where we assign one charge for every item $i$ with $r(i) \ge t^3$.
    We are allowed to assign these charges as follows: At most $5$ charges to every item $i$ with $R(i) \ge t$,
    at most $5$ charges to every early item, and at most $5$ charges to every bucket-epoch intersection that contains at least $t/3$ items.
    This budget guarantees that the total number of charges is at most $5|\{i : R(i) \ge t\}| + O\big(\frac{n}{t}\big)$, which shows the lemma.

    Let $i \in [n]$ be an index with $r(i) \ge t^3$.
    If $i$ is early, then we charge $i$ to itself.
    Otherwise, let $E$ be the epoch of $i$ and $E-1$ the previous epoch.
    Since $i$ is not early, at least $4t^2$ items were inserted into its bucket $B$ before $i$.
    At most $2t^2$ of these were inserted during epochs $E-1$ and $E$, so $B$ is $(E-1)$-heavy.
    If $R(i) \ge t$, then we charge $i$ to itself.
    Otherwise, consider the \emph{finger path} $j_1 = F(i), j_2 = F(j_1), \dots$
    until we hit the first item $j_a$ that is early, has $R(j_a) \ge t$, or until we reach $j_a$ for $a = t$.
    This finger path is very well-behaved: We have $i > j_1 > \dots > j_a$, and also $j_a > i - t^2$ since $R(i), R(j_1), \dots, R(j_{a-1}) < t$, and $a \le t$.
    In particular, the finger path lies within the epochs $E$ and $(E-1)$,
    and, since $r(i) \ge t^3$, it is disjoint from bucket $B$.

    Also, the finger path can never cross over an $(E-1)$-heavy bucket.
    To see this, put $\rho(j):=\rank_{X[1..j]}(X[j])$.
    For any two vertices $p,q$ on the finger path, including $i$, we have
    \[
    |p-q|<t^2
    \qquad\text{and}\qquad
    |\rho(p)-\rho(q)|<t^2.
    \]
    Suppose that $X[p]$ and $X[q]$ lie on opposite sides of an $(E-1)$-heavy bucket $C$.
    At least $2t^2$ elements of $C$ were inserted before epoch $E-1$.
    Since the insertion ranks can change by at most $|p-q|$ between steps $p$ and $q$, we have $
    |\rho(p)-\rho(q)|\ge 2t^2-|p-q|\ge t^2$,
    a contradiction.

    Since $j_1,\dots,j_{a-1}$ are not early, their buckets are all $(E-1)$-heavy.
    Thus, if $a>1$, the barrier argument shows that $j_1,\dots,j_{a-1}$ lie in the same $(E-1)$-heavy bucket $B'$ adjacent to $B$ among the $(E-1)$-heavy buckets,
    and $j_a$ lies in $B'$ or in an $(E-1)$-light bucket between $B$ and $B'$.
    If $a=1$, then $j_a$ lies in a neighbouring $(E-1)$-heavy bucket or in a light bucket between $B$ and the next heavy bucket.

    If $a<t$, then we charge $i$ to $j_a$.
    Otherwise, one of the pairs $(B',E)$ and $(B',E-1)$ contains at least half of $j_1,\dots,j_{t-1}$, and thus at least $t/3$ items, so we charge $i$ to that pair.

    To analyse the number of charges assigned to each item or pair,
    observe that every bucket-epoch pair contains at most one index $i \in [n]$ with $r(i) \ge t^3$.
    Indeed, two items in the same bucket and epoch have both index distance less than $t^2$ and value distance less than $t^3$.

    Consider an item $j$ that receives a charge as the endpoint of a finger path.
    If $j$ lies in epoch $K$, the source item $i$ lies in epoch $K$ or $K+1$.
    For each of these two possibilities, the source bucket $B$ must be one of the two $(E-1)$-heavy buckets neighbouring the bucket of $j$, or the light interval containing it.
    Since each bucket-epoch pair contains at most one source item, $j$ receives at most four such charges, plus possibly one charge from itself.

    Similarly, consider a bucket-epoch pair $(B',K)$ that receives a charge.
    The source item lies in epoch $K$ or $K+1$, and its bucket is one of the two neighbouring heavy buckets of $B'$ at the start of epoch $E-1$.
    Thus, this pair receives at most four charges.

    Finally, there are $O(n/t)$ early items and $O(n/t)$ bucket-epoch pairs containing at least $t/3$ items.
    Therefore, the total number of charges is at most
    $
    5|\{i:R(i)\ge t\}|+O(n/t)$,
    proving the lemma.
\end{proof}

\section{Loc is the Dynamic Finger Property, and Hist is the Working Set} \label{sect:BST_loc_hist}
In \cref{sect:discovery}, we showed that insertion sort using a dictionary with the unified bound is $\Reg$-adaptive.
In this section, we similarly relate $\Loc$ and $\Hist$ to the dynamic-finger and working-set properties of binary search trees~\cite{ProofreadCole2000,SleatorTarjan1985}.

Recall that a binary search tree stores a sorted set $Y$ subject to an $\textsc{access}(y)$ operation that forces it to navigate from the root to some item $y \in Y$.
A BST has the \emph{working-set property} if $\textsc{access}(y)$ takes amortised $O(\log(2+\age(y)))$ time, where $\age(y)$ counts the distinct keys accessed since the previous access to $y$. For a key not previously accessed, use $\age(y)=|Y|$. Counting all intervening access operations gives an upper bound on this age, which suffices in the proofs below.
Splay trees have the working-set property~\cite{SleatorTarjan1985}.
A BST has the \emph{dynamic-finger property}~\cite{ProofreadCole2000} if $\textsc{access}(y)$ takes amortised $O(\log(2+d(y',y)))$ time,
where $y'$ is the previously accessed item and $d(y', y)$ is the number of items in $Y$ that lie between $y$ and $y'$.
Cole~\cite{ProofreadCole2000}, building on joint work with Mishra, Schmidt, and Siegel~\cite{ProofreadColeEtAl2000}, proves the dynamic-finger property for splay trees.

There are different algorithms for inserting an item $y'$ into a binary search tree.
For this section, we allow any insertion algorithm that accesses both the predecessor and successor of $y'$ in $Y$,
adds a new node $y'$, and then accesses $y'$.
To perform insertion sort, we start with a binary search tree that contains two dummy elements $-\infty, \infty$,
and then insert $X[1], \dots, X[n]$ in this order.

\begin{theorem}
    Insertion sort with a binary search tree that has the dynamic finger property is $\Loc$-adaptive.
\end{theorem}
\begin{proof}
    Inserting $X[1]$ takes constant time.
    For $i \ge 2$, when inserting $y = X[i]$, the last accessed item was $y' = X[i-1]$.
    Let $z$ be the predecessor or successor of $y$ accessed first. Then $d(y',z)\leq|\rank_{X[1..i]}(X[i])-\rank_{X[1..i]}(X[i-1])|+1$, while the remaining accesses during insertion take constant amortised time.
    Thus, the total running time is
    \begin{align*}
    &O\Big(\sum_{i=2}^{n} (2 + \log(3+|\rank_{X[1..i]}(X[i]) - \rank_{X[1..i]}(X[i-1])|)\Big)\\
    =\ & O\Big(n + \sum_{i=2}^{n} \log(\max\{1, |\rank_{X[1..i]}(X[i]) - \rank_{X[1..i]}(X[i-1])|\})\Big) = O(n + \log \Loc). \qedhere
    \end{align*}
\end{proof}
\begin{theorem}
    Insertion sort with a binary search tree that has the working set property is $\Hist$-adaptive.
\end{theorem}
\begin{proof}
    Inserting $X[1]$ takes constant time.
    For $i \ge 2$, when inserting $y = X[i]$ into $Y$, let $j = \arg\min\left\{i-j:j<i \land \left|\operatorname{rank}_{X[1..i]}(X[i]) -\operatorname{rank}_{X[1..i]}(X[j])\right|\le 1\right\}$.
    Then $y' := X[j]$ is the predecessor or successor of $X[i]$ in $Y$, and the other predecessor or successor of $X[i]$ is an element $y''$ that was inserted before $X[j]$.
    Hence, the insertion of $y'$ accessed both $y'$ and $y''$.
    Since every insertion accesses three nodes, this shows $\age(y'), \age(y'') \le 3(i-j)$,
    so that accessing $y''$ and $y'$ during the insertion of $y$ takes $O(1 + \log (i-j))$ time.
    We assume that accessing the node $X[i]$ that was just created is free.
    Then, the total running time is
    \[
    O\!\left(n+\sum_{i=2}^{n}\log\min\{i-j:j<i\land|\rank_{X[1..i]}(X[i])-\rank_{X[1..i]}(X[j])|\leq1\}\right)=O(n+\log\Hist(X)).\qedhere
    \]
\end{proof}

\section{Synergistic Sorting}
\label{app:synergystic}

Barbay, Ochoa and Satti~\cite{BarbayOchoaSatti2017} introduce two sorting algorithms \textsc{DLMSort} and \textsc{QuickSynergySort}
that combine many adaptive properties. We show that their notion of adaptivity is stronger than $\nRuns$, $\Block$, and $\HMulti$, but weaker than $\loc$ if the input has no duplicates.
Moreover, we show that their notion is incomparable to $\Reg$, as it is weaker than $\loc$ when the input has no duplicates, but there does exist families of inputs on which their algorithm is better. 
To this end, we define the measure $\Syn$ that describes their algorithm.

\begin{definition} \label{def:syn}
    Let $X_1, \dots, X_k$ be the partition of $X$ into maximal contiguous non-decreasing sequences.
    We say a value $x \in [n]$ is \emph{free} if, for every index $j$ with $X[j] = x$, we have $\ol{X}[j] = \ol{X}[j-1] + 1$ and $\ol{X}[j] = \ol{X}[j+1]-1$.
    Note that, if $x$ is free, then $x$ occurs in at most one $X_i$.
    For every non-free value $x$, let $m_x$ be the number of $X_i$ that contain at least one copy of $x$.
    Let $\Syn = \prod\limits_{x: \text{ non-free}} \binom{k}{m_x}$ and $f_{\Syn}(X)  = n + \log \Syn(X)$.
\end{definition}

\begin{theorem}[Corollary 1 and Theorem 5 in \cite{BarbayOchoaSatti2017}]
    \textsc{DLMSort} and \textsc{QuickSynergySort} are $\Syn$-adaptive.
\end{theorem}

\begin{lemma}
    $\Syn \preceq \nRuns$, $\Syn \preceq \HMulti$ and $\Syn \preceq \Block$.
\end{lemma}
\begin{proof}
    We use the notation of \cref{def:syn}. Since 
    Since $\binom{k}{i}\leq(ek/i)^i$,    
    we have
    \[
    \log \Syn = \sum_{x: \text{ non-free}} \log \binom{k}{m_x} \le \sum_{x: \text{ non-free}} m_x \log \frac{e k}{m_x} \le \sum_{x: \text{ non-free}} m_x \log (e k) \le n \log (e k).
    \]
    Since $k = \nRuns$, this shows $n + \log \Syn(X) \in O(n + n \log \nRuns)$, i.e. $\Syn \preceq \nRuns$.

    For every value $x \in [n]$, let $m_x'$ denote the number of occurrences of $x$ in $X$.
    Since $y \mapsto y \log \frac{e n}{y}$ is an increasing function for $y \in [0, n]$, we have
    \[
    \log \Syn \le \sum_{x: \text{ non-free}} m_x \log \frac{e k}{m_x} \le \sum_{x: \text{ non-free}} m_x \log \frac{e n}{m_x} \le \sum_{x \in [n]} m_x' \log \frac{e n}{m_x'}\le n \log e + n \HMulti,
    \]
    with the convention ``$0 \log \frac{en}{0} = 0$''. This shows $\Syn \preceq \HMulti$.
    
    Finally, observe that $\sum_{x: \text{ non-free}} m_x \le 2\Block$ and $k \le \Block$. Hence
    \[
    \log \Syn \le \sum_{x: \text{ non-free}} m_x \log \frac{e k}{m_x} \le \sum_{x: \text{ non-free}} m_x \log (e k) \le n \log e + 2 \Block \log \Block,
    \]
    which shows $\Syn \preceq \Block$.
\end{proof}

\begin{lemma} \label{lemm:loc_syn}
    If $X$ contains no duplicates, then $\loc \preceq \Syn$
\end{lemma}
If $X$ contains duplicates, then $\loc \not \preceq \Syn$ since $\Syn \preceq \HMulti$ but $\loc \not\preceq \HMulti$.

\begin{proof}
    We use the notation of \cref{def:syn}.
    If $X$ contains no duplicates, then $m_x = 1$ for every non-free value $x$. Thus, if $n'$ denotes the number of non-free values,
    then
    \[
    \log \Syn = \sum_{x: \text{ non-free}} \log \binom{k}{m_x} = n' \log k.
    \]
    Consider a single maximal contiguous non-decreasing subsequence $X_i$ and suppose it is of the form $X[a_i], X[a_i + 1], \dots, X[b_i]$.
    Let $n_i'$ denote the number of non-free values in $X_i$.
    For all $x, y > 0$, we have $\log(x) = \log\frac{x}{y} + \log(y) \le \log\frac{x}{y} + y$.
    Therefore,
    \[
    \sum_{j = a_i+1}^{b_i} \log(X[j] - X[j-1]) = \sum_{\substack{j = a_i+1\\X[j] \text{ non-free}}}^{b_i} \log(X[j] - X[j-1]) \le \sum_{\substack{j=a_i+1\\ X[j] \text{non-free}}}^{b_i} (\log k + \frac{X[j] - X[j-1]}{k}) \le n_i' \log k + \frac{n}{k}.
    \]
    Also, $\log|X[a_i] - X[a_i - 1]| \le \log n$.
    Therefore, summing over all $X_i$ yields
    \[
    \log \loc = \sum_{j=2}^{n} \log|X[j] - X[j-1]| \le k \log n +  \sum_{i=1}^{k} (n_i' \log k + \frac{n}{k}) = k \log n + n' \log k + n. 
    \]
    We have $k \le n'$ since $X[a_i]$ is non-free for every $i$.
    By \cref{lemm:entropy_bounds}~(2), $k \log n \le 2 n + k \log k$.
    Therefore
    \[
    \log \loc \le k \log n + n' \log k + n \le 3 n + k \log k + n' \log k \le 3 n + 2 n' \log k \in O(n + \log \Syn).
    \]
    This shows $\loc \preceq \Syn$.
\end{proof}

\begin{lemma}\label{lemm:Reg_not_syn}
For inputs with duplicates, $\Reg\not\preceq\Syn$.
\end{lemma}
\begin{proof}
Let $k\geq3$, put $
S=(1,k+1,2k+1,\dots,(k-1)k+1)$, 
and consider
$
X=(1,2,\dots,k^2,\underbrace{S,\dots,S}_{k\text{ copies}})$.
The input has length $N=2k^2$ and exactly $k+1$ maximal
contiguous non-decreasing runs.

Every value in $S$ occurs in all $k+1$ runs, so its
contribution to $\Syn$ is
$\binom{k+1}{k+1}=1$.
Every other value occurs exactly once, in the initial run.
Such a value is free unless it immediately follows a value
of $S$, or is the final value $k^2$ of the initial run.
Indeed, all other singleton values have consecutive stable
ranks with both neighbours in $X$.
Thus, at most $k+1$ singleton values are non-free, each
contributing $\binom{k+1}{1}=k+1$.
Consequently,
\[
\log\Syn(X)\leq(k+1)\log(k+1),
\qquad
N+\log\Syn(X)\in\Theta(N).
\]

Now consider an occurrence $X[i]$ in one of the last
$k-1$ copies of $S$, and any earlier finger $X[j]$.
If $X[j]$ belongs to the initial run, then $i-j\geq k$.
Otherwise, both values belong to $S$.
If they are equal, their occurrences are at least $k$
positions apart.
If they differ, their prefix-rank distance is at least $k$,
since the initial run already contains every value in $[k^2]$.
Therefore every finger choice contributes a factor at least
$k$ to the corresponding factor of $\Reg(X)$.
There are $k(k-1)$ such occurrences, so
\[
\log\Reg(X)\geq k(k-1)\log k\in\Omega(N\log N).
\]
Hence $N+\log\Reg(X)$ is not
$O(N+\log\Syn(X))$, proving $\Reg\not\preceq\Syn$.
\end{proof}

\section{Open Problems}

We have not fully considered the adaptive heap properties from~\cite{Rutschmann2026}, and conjecture that they realise as a sorting algorithm an equally strong, dual variant of $\Reg$. Since our paper has not developed the tools for dualisation under inputs with duplicates, and $O(n + \log \reg) = O(n + \log \Reg)$ when there are no duplicates, we left this question open.

\paragraph*{HSMS and HSUS-adaptive sorting}

In this paper, we have described adaptive algorithms for most of the strongest measures in \cref{fig:hierarchy_proofs}:
insertion sort variants for $\Reg$, $\Hist$, and $\Loc$,
and \texttt{SkeletonSort} for $\reg$.
While a $\reg$-adaptive sorting algorithm is automatically $\HSMS$-adaptive and $\HSUS$-adaptive,
we are not aware of any $\HSMS$-adaptive or $\HSUS$-adaptive sorting algorithm that does not go via $\reg$ (or $\Reg$).
Since $\HSMS, \HSUS$ are considerably weaker notions than $\reg$, it should be easier to design an adaptive algorithm for them,
but we were unable to do so.




\bibliographystyle{plainurl}
\bibliography{refs}

\end{document}